\documentclass[12pt]{article}
\usepackage{amsthm,amsmath,amsfonts,amssymb,cases}
\newcommand{\R}{{\mathord{\mathbb R}}}
\newcommand{\Z}{{\mathord{\mathbb Z}}}
\newcommand{\N}{{\mathord{\mathbb N}}}
\newcommand{\C}{{\mathord{\mathbb C}}}

\def\chib {\overline{\chi}}

\newcommand{\HH}{\mathcal{H}}

\newcommand{\FF}{\mathcal{F}}

\newcommand{\WW}{\mathcal{W}}
\newcommand{\hh}{\mathfrak{h}}
\newcommand{\UU}{\mathcal{U}}

\newcommand{\umm}{\underline{m}}
\newcommand{\unn}{\underline{n}}
\newcommand{\upp}{\underline{p}}
\newcommand{\uqq}{\underline{q}}
\newcommand{\uzz}{\underline{0}}
\newcommand{\ujj}{\underline{j}}

\newcommand{\ran}{{\rm Ran}}

\newcommand{\un}[1]{\underline{#1}}

\DeclareMathOperator*{\esssup}{ess\,sup}

\newcommand{\ben}{\begin{displaymath}}
\newcommand{\een}{\end{displaymath}}
\newcommand{\beqn}{\begin{equation}}
\newcommand{\eeqn}{\end{equation}}
\newcommand{\beqna}{\begin{eqnarray*}}
\newcommand{\eeqna}{\end{eqnarray*}}

\def\inf{{\rm inf}\,}

\def\supp{\operatorname{supp}}

\newcommand{\sfrac}[2]{\textrm{\footnotesize $\frac{#1}{#2}$}}

\newtheorem{lemma}{Lemma}
\newtheorem{theorem}[lemma]{Theorem}
\newtheorem{remark}[lemma]{Remark}

\newtheorem{corollary}[lemma]{Corollary}
\newtheorem{definition}[lemma]{Definition}

\begin{document}
\title{
Smoothness and analyticity of perturbation expansions in QED}
\author{\vspace{5pt} D. Hasler$^1$\footnote{
E-mail: david.hasler@math.lmu.de \quad
On leave from College of William \& Mary} and I.
Herbst$^2$\footnote{E-mail: iwh@virginia.edu.} \\
\vspace{-4pt} \small{$1.$ Department of Mathematics,
Ludwig Maximilians University} \\ \small{Munich, Germany }\\
\vspace{-4pt}
\small{$2.$ Department of Mathematics, University of Virginia,} \\
\small{Charlottesville, VA, USA}\\}
\date{}
\maketitle

\begin{abstract}
We consider the ground state of an atom in the framework of non-relativistic qed.
We assume that the  ultraviolet cutoff is of the order of the  Rydberg energy and that the atomic
Hamiltonian has a non-degenerate ground state.
We show that the ground state energy and
 the ground state  are $k$-times continuously differentiable  functions of the fine structure constant and respectively
the square root of the fine structure constant on some nonempty interval $[0,c_k)$.
\end{abstract}

\section{Introduction}

Non-relativistic quantum electrodynamics (qed) is the theory describing the interactions between electrically charged non-relativistic 
quantum mechanical matter and the quantized electromagnetic field.
In this paper we  investigate expansions of the ground state and the ground state energy of an atom
as functions of the fine structure constant $\alpha$, as  $\alpha$ tends to zero.
In \cite{BFP06,BFP09} it was proven that there  exists  an asymptotic expansion involving
coefficients which depend on the coupling parameter $\alpha$ and have at most
mild singularities.
In \cite{BCVV09,HHS05,HS02}  related expansions of the ground state energy were obtained
and it was shown that  logarithmic divergences can occur in
non-relativistic qed. On the other hand
it was shown that an atom in a dipole
approximation of qed (which effectively leads to an infrared regularization) has a ground state and ground state energy
which are analytic functions of the coupling constant \cite{GH09}.

This paper can be viewed as  a continuation of  \cite{HH10-2}, where
 it was  shown that the ground state as
well as the ground state energy of the atom are analytic functions of the coupling constant, $g$,
which couples to the vector potential.  Moreover in   \cite{HH10-2} it was shown that in  an expansion in powers of
$g$, the corresponding expansion coefficients
are bounded as functions  of a  coupling constant, $\beta$,
which originates from the coupling  to  the electrostatic potential.
The main result of this paper  states  that these expansions coefficients are  $C^\infty$ functions of $\beta$,
and we obtain satisfactory bounds on the first  $k$ derivatives with respect to $\beta$.
We consider an atom which is coupled to the quantized radiation field in a scaling limit
where the ultraviolet cutoff is measured in units of Rydberg. This scaling limit is
a reasonable limit to study the properties of atoms. For example in this scaling limit estimates
on the lifetimes of metastable states \cite{HHH08,BFS99} were proven, which agree with experiment, see also \cite{AFFS09}.
Moreover, it was shown  \cite{GZ09}  that the ionization probability agrees with
calculations done by  physicists.
As a corollary of the main result of this paper, we show  that the ground state and the ground state energy have
convergent power series expansions, with
$\alpha$ dependent coefficients which are $C^\infty$ functions of $\alpha \geq  0$.
We show that the ground state energy as
well as the ground state  are $k$-times continuously differentiable  functions of $\alpha$
 respectively $\alpha^{1/2}$ on some nonempty interval $[0,c_k)$.
Moreover, it follows that the ground state as well as the ground state energy are given
as an asymptotic series in powers of $\alpha^{1/2}$ and $\alpha$, respectively, with constant coefficients.
These coefficients can be calculated by means of ordinary perturbation theory in a straight forward manner.
As a consequence of our result it follows that in the scaling limit  where the ultraviolet cutoff is of the order
of the Rydberg energy  no logarithmic terms occur. This  clarifies an issue which 
was raised in \cite{BFP09}, see the remark on Page 1031 therein.

Let us now address the proof of the main result.
It is well known that the  ground state energy is embedded in the continuous
spectrum.  In such a situation  regular perturbation theory is
typically not applicable and other methods have to be employed.
To prove the existence result as well as the analyticity result
 we use   a variant of the operator theoretic renormalization analysis as introduced
in \cite{BFS98}.
An important ingredient of the proof is that by rotation invariance one can
infer that in the renormalization analysis, terms which are linear in
creation and annihilation operators do not occur. This is explained in \cite{HH10-2}.
In that case it follows that the renormalization transformation is
a contraction even without infrared regularization. A similar
idea was used in a paper to prove existence and analyticity of the ground state and
ground state energy in the spin-boson model \cite{HH10-1}. In the proof we will use
results  obtained in \cite{HH10-1} and \cite{HH10-2}.
We note that similar ideas were used also in  \cite{GH09}.
The main new ingredient in the proof is the control of derivatives with respect to
the parameter $\beta$. The main estimates which control these derivatives
 are contained in Theorem \ref{ini:thm1b} and Lemma \ref{initial:thmE22}         for the initial Feshbach transformation and  in  Lemma \ref{codim:thm3first},   Theorem \ref{contk1}, and
Theorem \ref{thm:bcfsmain} (d) for  the renormalization transformation.
The most delicate  estimates are used in the  proof of   Lemma \ref{initial:thmE22} and Theorem \ref{contk1}, and
can be considered as the key ingredients of the proof.

\section{Model and Statement of Results}

Let $(\hh, \langle \cdot , \cdot \rangle_\hh)$ be a
Hilbert space. We introduce the direct sum of the $n$-fold tensor product of $\hh$ and set
$$
\mathcal{F}(\hh) := \bigoplus_{n=0}^\infty \mathcal{F}^{(n)}(\hh) ,  \qquad \mathcal{F}^{(n)}(\hh) = \hh^{\otimes^n} ,
$$
where we have set $\hh^{\otimes 0} := \C$.
We introduce the vacuum vector $\Omega := (1,0,0,...) \in \FF(\hh)$.
The space  $\mathcal{F}(\hh)$ is an inner product space
where the inner product is induced from the inner product in $\hh$. That is, on vectors $\eta_1 \otimes \cdots \eta_n,  \varphi_1 \otimes \cdots \varphi_n  \in \FF^{(n)}(\hh)$
we have
$$
\langle \eta_1 \otimes \cdots \eta_n , \varphi_1 \otimes \cdots \varphi_n  \rangle :=  \prod_{i=1}^n
\langle \eta_{i} , \varphi_{i} \rangle_\hh .
$$
This definition extends to all of $\FF(\hh)$ by bilinearity and continuity. We introduce the bosonic Fock space
$$
\FF_s(\hh) := \bigoplus_{n=0}^\infty \mathcal{F}_s^{(n)}(\hh)  , \qquad \mathcal{F}_s^{(n)}(\hh) := S_n \FF^{(n)}(\hh) ,
$$
where $S_n$  denotes the orthogonal projection onto the subspace of totally symmetric
tensors in  $\FF^{(n)}(\hh)$. For $h \in \hh$ we introduce the so called creation operator $a^*(h)$ in $\FF_s(\hh)$
which is defined on vectors  $\eta \in \FF^{(n)}_s(\hh)$
by
\beqn  \label{eq:formala}
a^*(h) \eta := \sqrt{n+1} S_{n+1} ( h \otimes \eta ) \; .
\eeqn
The operator $a^*(h)$ extends by linearity to a densely defined  linear operator on $\FF(\hh)$.
One can show that $a^*(h)$ is closable, c.f. \cite{reesim2}, and we denote its closure by the same symbol.
We introduce the annihilation operator by  $a(h) := (a^*(h))^*$.
For a closed operator $A \in \hh$ with domain $D(A)$ we introduce the operator $\Gamma(A)$ and $d\Gamma(A)$ in $\FF(\hh)$ defined
on vectors $\eta = \eta_1 \otimes \cdots \otimes \eta_n \in \FF^{(n)}(\hh)$, with $\eta_i \in D(A)$, by
$$
\Gamma(A) \eta =  A \eta_1 \otimes \cdots \otimes A  \eta_n
$$
and
$$
d \Gamma(A) \eta = \sum_{i=1}^n  \eta_1 \otimes \cdots \otimes \eta_{i-1} \otimes A \eta_i \otimes \eta_{i+1} \otimes \cdots \otimes \eta_n
$$
and extended by linearity to  a densely defined linear operator on $\FF(\hh)$. One can show
that  $d \Gamma(A)$ and $\Gamma(A)$ are closable, c.f. \cite{reesim2}, and we denote their closure by the same symbol.
The operators $\Gamma(A)$ and $d \Gamma(A)$ leave the subspace $\FF_s(\hh)$ invariant, that is,
their restriction to $\FF_s(\hh)$ is densely defined, closed,  and has range contained in $\FF_s(\hh)$.
To define qed, we fix
$$
\hh := L^2( \R^3 \times \Z_2 )
$$
and set $\FF := \FF_s(\hh)$.  We define the operator of the free field energy
by
$$
H_f := d \Gamma(M_\omega) ,
$$
where $\omega(k,\lambda) := \omega(k) := |k|$ and $M_\varphi$ denotes the operator of multiplication with the function $\varphi$.
For $f \in \hh$ we write
$$
a^*(f) = \sum_{\lambda=1,2} \int f(k,\lambda) a^*(k,\lambda)  , \qquad a(f) = \sum_{\lambda=1,2} \int \overline{f(k,\lambda)} a^*(k,\lambda) .
$$
where $a(k,\lambda)$ and $a^*(k,\lambda)$ are operator-valued distributions.
They satisfy the
following commutation relations, which are to be understood in the sense of distributions,
$$
[a(k,\lambda), a^*(k',\lambda') ] = \delta_{\lambda \lambda'} \delta(k - k')   , \qquad [a^\#(k,\lambda), a^\#(k',{\lambda'}) ]  = 0 \; ,
$$
where $a^{\#}$ stands for  $a$ or $a^*$. For $\lambda=1,2$ we introduce the so called polarization vectors
$$
\varepsilon(\cdot , \lambda) : S^2 := \{ k \in \R^3 | |k| = 1 \} \to \R^3
$$
to be measurable maps such that for each $k \in S^2$ the vectors $\varepsilon(k,1), \varepsilon(k,2),k$ form an orthonormal basis of $\R^3$.
We extend $\varepsilon(\cdot , \lambda)$ to $\R^3 \setminus \{ 0 \}$ by setting
$
\varepsilon(k,\lambda) := \varepsilon(k/|k|,\lambda)
$
for all nonzero $k$.  For $x \in \R^3$ we define the field operator
\beqn \label{eq:afield}
A_{\Lambda}(x) =  \sum_{\lambda=1,2} \int \frac{dk \kappa_{\Lambda}(k)}{\sqrt{2 |k|}} \left[ e^{-ik \cdot x}
\varepsilon(k,\lambda) a^*(k,\lambda)  + e^{ik \cdot x} \varepsilon(k, \lambda) a(k, \lambda) \right] \ ,
\eeqn
where the function  $\kappa_{\Lambda}$  serves as a cutoff, which satisfies
$\kappa_{\Lambda}(k) = 1$ if $ |k| \leq \Lambda$  and which
is zero otherwise. $\Lambda > 0$ is an ultraviolet cutoff, which we assume to be finite.
Next we introduce the atomic Hilbert space, which describes the configuration of
$N$ electrons, by
$$
\HH_{\rm at}  := \{ \psi \in L^2(\R^{3N}) | \psi(x_{\sigma(1)},...,x_{\sigma(N)}) =
{{\rm sgn}(\sigma)} \psi(x_1,...,x_N) ,  \sigma \in  \mathfrak{S}_N \}   ,
$$
where $\mathfrak{S}_N$ denotes the group of permutations of $N$ elements, ${\rm sgn}$ denotes the signum of the permutation,
and  $x_j \in \R^3$ denotes the coordinate of the $j$-th electron.
We will consider the following operator in $\HH := \HH_{\rm at} \otimes \FF$,
\begin{equation} \label{eq:hamiltoniandefinition}
H_{g,\beta}  =  \ \ : \sum_{j=1}^N  ( p_j + g A_{\Lambda}(\beta x_j) )^2 :  +  V   + H_f ,
\end{equation}
where
$p_j = - i \partial_{x_j}$,
$V = V(x_1,...,x_N)$ denotes the potential, and  $:( \, \cdot \, ):$ stands for the Wick product.
We will make the following assumptions on the potential $V$,  which are related to the atomic Hamiltonian
$$
H_{\rm at} := - \Delta   + V  ,
$$
which acts in $\HH_{\rm at}$.  We introduced the Laplacian $-\Delta := \sum_{j=1}^N p_j^2$.

\vspace{0.5cm}
\noindent
{\bf Hypothesis (H)} The potential $V$ satisfies the following properties:
\begin{itemize}
\item[(i)] $V$ is symmetric under permutations and invariant under rotations.
\item[(ii)] $V$ is infinitesimally operator bounded with respect to $-\Delta$.
\item[(iii)] $E_{\rm at} := \inf \sigma(H_{\rm at} )$ is a non-degenerate isolated eigenvalue of $H_{\rm at}$.
\end{itemize}

\vspace{0.5cm}

Note that for the Hydrogen, $N=1$, the potential $V(x_1) = - {|x_1|}^{-1}$
satisfies Hypothesis  (H). Moreover (ii) of Hypothesis (H) implies that
$H_{g,\beta}$ is a self-adjoint operator with domain $D(-\Delta + H_f)$ and that $H_{g,\beta}$ is essentially
self adjoint on any operator core for $-\Delta +H_f$,
 see for example \cite{H02,HH08}. For a precise definition of the  operator  in    \eqref{eq:hamiltoniandefinition}, see  Appendix A.
We  will use the notation $D_r(w) := \{ z \in \C | |z - w| < r \}$ and $D_r := D_r(0)$.
Let us now state the main result of the paper.

\begin{theorem} \label{thm:main1}   Assume Hypothesis (H) and let $k \in \N_0$. Then there exists
a positive constant $g_0$ such that for all  $g \in D_{g_0}$ and  $ \beta \in \R$
the operator $H_{g,\beta}$ has an eigenvalue $E_{\beta}(g)$ with eigenvector $\psi_{\beta}(g)$ and eigen-projection
$P_{\beta}(g)$ satisfying the following properties.
\begin{itemize}
\item[(i)] For $g \in \R \cap D_{g_0}$ we have  $E_{\beta}(g) = {\rm inf}  \sigma ( H_{g,\beta}) $, and
 for all $g \in D_{g_0}$ we have $P_{\beta}(g)^* = P_{\beta}(\overline{g})$.
\item[(ii)] $g \mapsto  E_{(\cdot)}(g)$,  $g \mapsto  \psi_{(\cdot)}(g)$, and
 $g \mapsto  P_{(\cdot)}(g)$ are analytic  functions on $D_{g_0}$ with values in $C_B^k(\R)$,
 $C^k_B(\R;\HH)$, and $C^k_B(\R;\mathcal{B}(\HH))$, respectively.
 \item[(iii)] There exists a finite and positive $C$ such that for all $g \in D_{g_0}$ we have
$$
\| E_{(\cdot)}(g) \|_{C^k(\R)} \leq C  , \quad \| \psi_{(\cdot)}(g) \|_{C^k(\R;\HH)} \leq C , \quad
\| P_{(\cdot)}(g) \|_{C^k(\R;\mathcal{B}(\HH))} \leq C   .
$$
\end{itemize}
\end{theorem}

The next result states that the expansions coefficients of the eigenvalue, eigenfunction, and the corresponding eigenprojection
are $C^\infty$ as functions of  $\beta$.

\begin{corollary} \label{thm:main2}
Assume Hypothesis (H) and let $k \in \N_0$. Then there exists
a positive constant $g_0$ such that for all  $g \in D_{g_0}$ and  $ \beta \in \R$
the operator $H_{g,\beta}$ has an eigenvalue $E_{\beta}(g)$ with eigenvector $\psi_{\beta}(g)$ and eigen-projection
$P_{\beta}(g)$ satisfying the following properties.
On $D_{g_0}$ we have the convergent expansions
\begin{equation} \label{eq:expansion1}
E_{\beta}(g) = \sum_{n=0}^\infty E^{(2n)}_{\beta} g^{2n}, \quad
\psi_{\beta}(g) = \sum_{n=0}^\infty \psi_{\beta}^{(n)} g^n  , \quad P_{\beta}(g) = \sum_{n=0}^\infty P^{(n)}_{\beta} g^n .
\end{equation}
There exist  finite and positive constants $C$ and $r$ such that
$$
\| E_{(\cdot)}^{(2n)} \|_{C^k(\R)} \leq C r^{2n} , \quad \| \psi_{(\cdot)}^{(n)} \|_{C^k(\R;\HH)} \leq C r^n , \quad
\| P_{(\cdot)}^{(n)} \|_{C^k(\R;\mathcal{B}(\HH))} \leq C r^n  .
$$
The expansion coefficients are as functions of $\beta$ in $C^\infty(\R)$, $C^\infty(\R;\HH)$, and $C^\infty(\R;\mathcal{B}(\HH))$, respectively.
\end{corollary}

Various conclusions can be drawn from Theorem  \ref{thm:main1}. For instance, if we set
$\beta = \alpha \geq 0$ and $g = \alpha^{3/2}$ then we  obtain the following corollary.
It  states that the ground state and the ground state energy of an atom in qed, in a  scaling limit where the ultraviolet cutoff is
of the order of the  Rydberg energy, can
be differentiated arbitrarily many times as functions of $\alpha$ and $\alpha^{1/2}$, respectively,
provided one chooses $\alpha$ sufficiently small (depending on the number of derivatives).
As a  conclusion it follows  that no logarithmic terms
appear in this scaling limit.

\begin{corollary} \label{cor:main1}  Assume Hypothesis (H). There exists a positive  $\alpha_0$ such that for $ 0 \leq \alpha \leq \alpha_0$ the operator
$H_{\alpha^{3/2},\alpha}$ has a ground state $\psi(\alpha^{1/2})$ with ground
state energy $E(\alpha)$ such that we have the convergent expansions on $[0,\alpha_0)$
\begin{equation} \label{eq:expansion}
E(\alpha) = \sum_{n=0}^\infty E^{(2n)}_\alpha \alpha^{3 n }  ,\qquad
\psi(\alpha^{1/2}) = \sum_{n=0}^\infty \psi_\alpha^{(n)} \alpha^{3n/2} .
\end{equation}
The coefficients   $E_\alpha^{(n)}$ and $\psi_\alpha^{(n)}$ are as functions of $\alpha$
in $C^\infty([0,\infty))$ and $C^\infty([0,\infty);\HH)$, respectively.
For every $k \in \N_0$ there exists a positive  $\alpha_0^{(k)}$ such that
$\psi(\cdot)$ and $E(\cdot)$ are $k$-times continuously differentiable on $[0,\alpha_0^{(k)})$.
\end{corollary}
In \cite{BFP06,BFP09} it was shown  that there exist
coefficients of the type  \eqref{eq:expansion}   which   have
slower growth than $\alpha^{-t}$ for any $t > 0$. Corollary  \ref{cor:main1}
states that  the  coefficients   $E_\alpha^{(n)}$ and $\psi_\alpha^{(n)}$ are   in fact smooth.
Let us note that  Corollary  \ref{cor:main1} implies the following corollary which
states that the ground state and the ground state energy can be written in terms of an asymptotic
series with constant coefficients in the sense of \cite{reesim4}.

\begin{corollary} \label{cor2}  Assume Hypothesis (H). There exist formal power series with constant coefficients
$\sum_{n=0}^\infty c^{(n)} \alpha^{n/2}$ and $
\sum_{n=0}^\infty e^{(n)} \alpha^{n}  $
which  are asymptotic to the
ground state and
the ground state energy  of
$H_{\alpha^{3/2},\alpha}$  as $\alpha \downarrow 0$, respectively.
\end{corollary}

In view of  Corollary  \ref{cor:main1} and the continuity in the infrared cutoff which has been established  in \cite{HH10-2} one can calculate
$c^{(n)}$ and $ e^{(n)}$ of Corollary
\ref{cor2} using for example ordinary Rayleigh Schr\"odinger perturbation theory to determine first $\psi_\alpha^{(n)}$
and  $E^{(2n)}_\alpha$, in Eq.  \eqref{eq:expansion}, and then using a Taylor expansion of these coefficients.

\section{Outline  of the Proof}
\label{sec:outline}

The main method used in  the proof of Theorem  \ref{thm:main1}   is
operator theoretic renormalization \cite{BFS98,BCFS03} and the fact that renormalization preserves
analyticity \cite{GH09,HH10-1}. The renormalization procedure is an iterated application of
the so called smooth Feshbach map. The smooth Feshbach map is reviewed in Appendix  B and necessary properties of
it are summarized.
In this paper we will use many results stated in  the  previous papers \cite{HH10-1} and \cite{HH10-2}.
The generalization from the Fock space over $L^2(\R^3)$, as considered in \cite{HH10-1}, to a Fock space
over $L^2(\R^3 \times \Z_2)$ is straight forward. To be able to show that  the renormalization transformation
is a suitable contraction we use a   rotation invariance argument, as explained in \cite{HH10-2}.
The main new ingredient is to control  derivatives with respect the $\beta$. The subtleties
originate from the reparameterization of the spectral parameter
In Section \ref{sec:symmetries} we define an   $SO(3)$ action
on the atomic Hilbert space and the Fock space, which leaves the Hamiltonian invariant.
In Section \ref{sec:ban} we introduce spaces which are needed
to define the renormalization transformation.
In Section \ref{sec:ini} we show that after an initial Feshbach transformation the Feshbach map is
in a suitable Banach space.
This allows us to perform a renormalization analysis, which is
the content of  Section \ref{sec:ren:def}. We use
 results from   \cite{HH10-1} and complement it with new estimates needed to
 control differentiation with respect to $\beta$.
In Section \ref{app:contraction} we prove the contraction property of the
renormalization transformation.
 In Section \ref{sec:prov} we put  the pieces together and
prove  Theorem \ref{thm:main1}. The proof is based  on Theorems  \ref{thm:inimain1}  and   \ref{thm:bcfsmain}.

We use the notation  $\R_+ = [0, \infty)$. For a multi-index $\umm \in \N_0^l$ we use the usual convention  $|\umm| = \sum_{i=1}^l m_i $ and
$\umm!= \prod_{i=1}^l ( m_i!)$.
We shall make repeated use of the so called pull-through formula which is given in  Lemma \ref{lem:pullthrough}, in Appendix A.
We refer the reader to the appendix for notation of  function spaces and will use
Lemma \ref{lem:weakstronanalyt1}.
Finally, let us note that using  an appropriate
scaling we can assume without loss of generality that the distance between the lowest
eigenvalue of $H_{\rm at}$ and the rest of the spectrum is one, i.e.,
\begin{equation} \label{eq:hatscale}
E_{\rm at,1} - E_{\rm at} = 1  ,
\end{equation}
where $E_{\rm at,1} := \inf \left\{ \sigma( H_{\rm at}) \setminus\{ E_{\rm at}\} \right\}$. Any Hamiltonian of the form
\eqref{eq:hamiltoniandefinition} satisfying Hypothesis (H)  is up to a positive multiple unitarily equivalent to an operator
satisfying  \eqref{eq:hatscale} and Hypothesis (H), but with a rescaled potential and with different values for $\Lambda$, $\beta$,
and $g$, see \cite{HH10-2}.

\section{Symmetries}

\label{sec:symmetries}

Let us introduce  a representation of $SO(3)$ on $\HH_{\rm at}$ and $\hh$. For details  see \cite{HH10-2}. For $R \in SO(3)$ and $\psi \in \HH_{\rm at}$ we define
$$
\mathcal{U}_{\rm at}(R) \psi(x_1,...,x_N) = \psi(R^{-1} x_1, ... , R^{-1} x_N) .
$$
To define an $SO(3)$ action on Fock space it is convenient to consider a different but equivalent representation of the
Hilbert space $\hh$.
We introduce the Hilbert space $\hh_0 := L^2(\R^3 ; \C^3)$. We consider the
subspace   of transversal vector fields
$$
\hh_T := \{ f \in \hh_0 | k \cdot f(k)  = 0  \} .
$$
It is straightforward to verify that  the map
$\phi : \hh  \to \hh_T $ defined by
\begin{eqnarray*}
 (\phi f)(k) :=  \sum_{\lambda=1,2} f(k,\lambda) \varepsilon(k,\lambda)
\end{eqnarray*}
establishes a unitary isomorphism with inverse
$$
(\phi^{-1}f)(k,\lambda) =   f(k) \cdot  \varepsilon(k,\lambda)  .
$$
We define the action of $SO(3)$ on $\hh_T$ by
$$
(\mathcal{U}_T(R) f )(k) = R f(R^{-1} k) .
$$
The function $R \mapsto \phi^{-1} \mathcal{U}_T(R) \phi$ defines a representation of $SO(3)$ on $\hh$ which we denote
by  $\UU_\hh$. This yields a representation on Fock space which we  denote  by $\UU_\FF$. It is characterized by
$$
\UU_\FF(R) a^\#(f) \UU_\FF(R)^* = a^\#(\UU_\hh(R) f) \quad , \quad \UU_\FF(R) \Omega = \Omega .
$$
It is straight forward to show that the Hamiltonian $H_{g,\beta}$ is $SO(3)$ invariant.

\section{Banach Spaces of Hamiltonians}
\label{sec:ban}

In this section we introduce Banach spaces of integral kernels, which
parameterize  certain subspaces of the space of bounded operators on Fock space.
These spaces are used to control the renormalization transformation. Then we
introduce Banach spaces, which we call extended Banach spaces, which are used to  control derivatives with respect to $\beta$.

The renormalization transformation will be defined on operators acting on the reduced Fock space
$\mathcal{H}_{\rm red}:= P_{\rm red} \FF$,
where we introduced the notation  $P_{\rm red}:=  \chi_{[0,1]}(H_f)$.
We will investigate bounded operators in $\mathcal{B}(\mathcal{H}_{\rm red})$ of the form
\beqn \label{eq:sum}
H(w) := \sum_{m+n \geq 0} H_{m,n}(w) ,
\eeqn
with
\begin{align}
& H_{m,n}(w) := H_{m,n}(w_{m,n}) ,   \nonumber \\
& H_{m,n}(w_{m,n}) := P_{\rm red} \int_{\un{B}_1^{m+n}} \frac{ d \mu( {K}^{(m,n)})}{|{K}^{(m,n)}|^{1/2}} a^*({K}^{(m)}) w_{m,n}(H_f, {K}^{(m,n)}) a(\widetilde{{K}}^{(n)}) P_{\rm red}  , \quad m+n \geq 1  ,  \label{eq:defhmn11} \\
& H_{0,0}(w_{0,0}) := w_{0,0}(H_f) , \nonumber
\end{align}
where $w_{m,n} \in L^\infty([0,1] \times \un{B}_1^m \times \un{B}_1^n)$
is an integral kernel for $m+n \geq 1$, $w_{0,0} \in L^\infty([0,1])$, and $w$ denotes the sequence of
integral kernels $(w_{m,n})_{m,n \in \N_0^2}$.
We have used and will henceforth use the following notation.  We set $K = (k, \lambda ) \in  \R^3 \times \Z_2$, and write
\begin{align*}
& \un{X} := X \times \Z_2  \quad , \quad B_1 := \{ x \in \R^3 | |x|< 1 \} \\
& K^{(m)} := ({K}_1, ... ,{K}_m ) \in  \left( {\R}^{3} \times \Z_2 \right)^m  ,
\quad \widetilde{{K}}^{(n)} := (\widetilde{{K}}_1, ... , \widetilde{{K}}_n ) \in \left( {\R}^{3} \times \Z_2 \right)^n , \\
&  {K}^{(m,n)}  := ({K}^{(m)}, \widetilde{{K}}^{(n)})  \\
& \int_{\un{X}^{m+n}} d  {K}^{(m,n)}  :=  \int_{{X}^{m+n}} \sum_{(\lambda_1,...,\lambda_m,\widetilde{\lambda}_1,...,\widetilde{\lambda}_n) \in \Z_2^{m+n} } dk^{(m)} d\widetilde{k}^{(n)} \\
&  dk^{(m)} := \prod_{i=1}^m {d^3 k_i} , \quad d\widetilde{k}^{(n)}  := \prod_{j=1}^n {d^3 \widetilde{k}_j}  ,  \quad
d K^{(m)} := d K^{(m,0)} , \quad  d \widetilde{K}^{(n)} := d K^{(0,n)} ,  \\
& d \mu (K^{(m,n)}) := (8 \pi )^{-\frac{{m+n}}{2}} d K^{(m,n)} \\
& a^*({K}^{(m)}) :=  \prod_{i=1}^m a^*({K}_i) , \quad a(\widetilde{{K}}^{(m)}) :=  \prod_{j=1}^m a(\widetilde{{K}}_j) \\
& | {K}^{(m,n)}| := | {K}^{(m)} | \cdot | \widetilde{{K}}^{(n)}| , \quad | {K}^{(m)} | := |k_1| \cdots |k_m | , \quad  | \widetilde{{K}}^{(m)} | := |\widetilde{k}_1| \cdots |\widetilde{k}_m | , \\
& \Sigma[{K}^{(m)}] := \sum_{i=1}^n |k_m |  \; .
\end{align*}
Note that in view of the pull-through formula  \eqref{eq:defhmn11} is equal to
\beqn \label{eq:defintegralkernel}
\int_{\underline{B}_1^{m+n}} \frac{ d \mu(K^{(m,n)})}{|K^{(m,n)}|^{1/2}} a^*(K^{(m)})  \chi(H_f +  \Sigma[K^{(m)}] \leq 1 ) w_{m,n}(H_f ; K^{(m,n)})
\chi(H_f +  \Sigma[\tilde{K}^{(n)}] \leq 1) a(\tilde{K}^{(n)} ) \; .
\eeqn
Thus we can restrict attention to integral kernels $w_{m,n}$ which are  essentially supported on the sets
\begin{eqnarray*}
\underline{Q}_{m,n} &:=& \{ ( r , K^{(m,n)}) \in [0,1] \times \underline{B}_1^{m+n}  \ | \ r  \leq 1 -
\max(\Sigma[K^{(m)}],
\Sigma[\widetilde{K}^{(m)}]) \} , \quad m + n \geq 1 .
\end{eqnarray*}
Moreover, note that  integral kernels can always be assumed to be symmetric. That is, they lie in the range of the symmetrization operator,
which is defined as follows,
\begin{eqnarray} \label{eq:symmetrization}
w_{M,N}^{({\rm sym})}(r;K^{(M,N)}) := \frac{1}{N!M!} \sum_{\pi \in S_M} \sum_{\widetilde{\pi} \in S_N} {w}_{M,N}(r,
K_{\pi(1)},\ldots,K_{\pi(N)}, \widetilde{K}_{\widetilde{\pi}(1)},\ldots,\widetilde{K}_{\widetilde{\pi}(M)}).
\end{eqnarray}

Note that \eqref{eq:defhmn11} is understood in the sense of forms. It defines a densely defined form
which can be seen to be bounded using   Lemma \ref{kernelopestimate}.
Thus it uniquely determines a bounded operator which we denote by $H_{m,n}(w_{m,n})$. This is explained in more
detail in  Appendix A. We have the following lemma.
\begin{lemma} For $w_{m,n} \in L^\infty([0,1] \times \un{B}_1^m \times \un{B}_1^n)$\label{lem:operatornormestimates} we have
\beqn \label{eq:operatornormestimate1}
\|H_{m,n}(w_{m,n}) \|_{} \leq \| w_{m,n} \|_{\infty} ( n! m!)^{-1/2} \; .
\eeqn
\end{lemma}
The proof follows using Lemma \ref{kernelopestimate}
and the estimate
\begin{equation} \label{eq:intofwKminus2}
 \int_{\underline{S}_{m,n}} \frac{d K^{(m,n)}}{|K^{(m,n)}|^{2}} \leq \frac{(8 \pi)^{m+n}}{{n! m!}} ,
\end{equation}
where $\underline{S}_{m,n} := \{ (K^{(m)},\widetilde{K}^{(n)}) \in \underline{B}_1^{m+n}
 \ | \Sigma[K^{(m)}] \leq 1 , \Sigma[\widetilde{K}^{(n)}] \leq 1 \}$.
The renormalization procedure will involve kernels which lie in the following Banach spaces.
We denote the norm of the Banach space  $L^\infty(\underline{B}_1^{m+n}; C[0,1])$ by
$\| \cdot \|_{\underline{\infty}}$.
We shall identify the space $L^\infty(\underline{B}_1^{m+n}; C[0,1])$ with a subspace of $L^\infty([0,1]\times \underline{B}_1^{m+n})$ by
setting
$$
w_{m,n}(r,K^{(m,n)}) := w_{m,n}(K^{(m,n)})(r) .
$$
This identification is used  for example in  (i) and (ii)  of  Definition \ref{def:wgartenhaag}.

\begin{definition} \label{def:wgartenhaag}
We define $\WW_{m,n}^\#$ to be the Banach space consisting of functions $w_{m,n} \in
L^\infty(\underline{B}_1^{m+n};C^1[0,1])$ satisfying the following properties:
\begin{itemize}
\item[(i)] $ w_{m,n} (1 - \chi_{\underline{Q}_{m,n}} ) = 0$, for $m + n \geq 1$,
\item[(ii)] $w_{m,n}(r, K^{(m)}, \widetilde{K}^{(n)})$ is totally
symmetric in the variables $K^{(m)}$ and $\widetilde{K}^{(n)}$
\item[(iii)] the following norm is finite
$$
\| w_{m,n} \|^\# := \| w_{m,n} \|_{\underline{\infty}} + \| \partial_r w_{m,n} \|_{\underline{\infty}} .
$$
\end{itemize}
For $0<\xi < 1$, we define the Banach space
$$
 \mathcal{W}^\#_{\xi} := \bigoplus_{(m,n) \in \N_0^2 } \mathcal{W}_{m,n}^\# \
$$
to consist of all sequences $w =( w_{m,n})_{m,n \in \N_0}$ satisfying
$$
\| w \|_\xi^\# := \sum_{(m,n)\in \N_0^2} \xi^{-(m+n)} \| w_{m,n}\|^\# < \infty .
$$
\end{definition}

Given $w \in \mathcal{W}_\xi^\#$, we
write $w_{\geq r}$ for the vector in $\mathcal{W}_\xi^\#$ given by
$$
(w_{\geq r})_{m + n} = \left\{ \begin{array}{ll} w_{m,n} & , \quad {\rm if} \ m+n \geq r \\ 0 & , \quad {\rm otherwise} . \end{array} \right.
$$
For $w \in \mathcal{W}^\#_{\xi}$, it is easy to see using \eqref{eq:operatornormestimate1} that
$
H(w) := \sum_{m,n} H_{m,n}(w)
$
converges in operator norm with bounds
\begin{align} \label{eq:opestimatgeq12}
 \| H(w_{\geq r} ) \|_{} \leq \xi^r \| w_{\geq r} \|_\xi^\# .
\end{align}
We shall use the notation
$$
W[w] := \sum_{m+n \geq 1} H_{m,n}(w) .
$$
We will use the following theorem, which is a straightforward generalization of a theorem proven in \cite{BCFS03}. A proof
can also be found in \cite{HH10-1}.

\begin{theorem} \label{thm:injective} The map $H : \WW_\xi^\# \to \mathcal{B}(\HH_{\rm red})$ is injective and bounded. Moreover
$\| H(w) \|_{} \leq \| w\|_\xi^\# $.
\end{theorem}

The integral kernels depend on the spectral parameter. To accommodate
for this we introduce the  Banach space $ \WW_\xi^{}  := C^\omega_B(D_{1/2}; \mathcal{W}^\#_{\xi})$
with norm
\begin{align*}
 \| w \|_\xi^{}  :=  \sup_{z \in D_{1/2}} \| w(z)  \|_\xi^{\#}
\end{align*}
Moreover, the integral kernels depend  on the coupling constants. We introduce the following Banach space
\begin{align*}
& \WW_{\xi}^{(k)}(S) := C^{\omega,k}_B(S \times \R ; \WW_\xi^\#) ,
\end{align*}
with the  norm
\begin{align*}
 \| w \|_{\xi,S}^{(k)}  &:= \sup_{(s,\beta) \in S \times \R} \sum_{m,n} \xi^{-m-n}
\max_{0 \leq l \leq k} \| \partial_\beta^l w(\beta,s)_{m,n} \|^\# .
\end{align*}
Observe that this norm is different but equivalent to the natural norm,
\begin{eqnarray*}
\max_{0 \leq l \leq k}  \sup_{(s,\beta) \in S \times \R} \sum_{m,n} \xi^{-m-n}
  \| \partial_\beta^l w(\beta,s)_{m,n} \|^\#  \leq
\| w \|_{\xi,S}^{(k)}  \leq  k  \max_{0 \leq l \leq k}      \sup_{(s,\beta) \in S \times \R} \sum_{m,n} \xi^{-m-n}
 \| \partial_\beta^l w(\beta,s)_{m,n} \|^\# .
\end{eqnarray*}
For notational compactness we will use an abbreviation for the case $S = D_{1/2}$ and set
 $\WW_{\xi}^{(k)} :=   \WW_{\xi}^{(\omega,k)}(D_{1/2}) $ and
$ \| \cdot \|_{\xi}^{( k)} :=    \| \cdot  \|_{\xi,S}^{(k)}   $.
We introduce the Banach space
\begin{align*}
& \WW_{\xi}^{(\#,k)} := C^{k}_B( \R ; \WW_\xi^\#)  , \quad   \| \cdot \|_{\xi}^{(\#, k)}    .
\end{align*}
with the norm
\begin{align*}
 \| w \|_{\xi}^{(\#,k)}  &:= \sup_{\beta \in \R} \sum_{m,n} \xi^{-m-n}
\max_{0 \leq l \leq k} \| \partial_\beta^l w(\beta)_{m,n} \|^\# .
\end{align*}

For $w \in \WW_\xi$ we will use the notation $w_{m,n}(z, \cdot) := (w_{m,n}(z))(\cdot)$.
We extend the definition of $H(\cdot)$ to $\WW_\xi$ in the natural way: for $w \in \WW_\xi$, we set
$$
\left( H(w) \right) (z) := H(w(z))
$$
and likewise for $H_{m,n}(\cdot)$ and $W[\cdot]$.
We say that a kernel $w \in \WW_\xi$ is symmetric
 if $w_{m,n}(\overline{z}) = \overline{w_{n,m}(z)}$  for all $z \in D_{1/2}$. Note that because of Theorem
\ref{thm:injective} we have the following lemma.
\begin{lemma} \label{lem:symmetry} Let $ w \in \WW_\xi$. Then
 $w$ is symmetric if and only if $H(w(\overline{z}))= H(w(z))^*$ for all $z \in D_{1/2}$.
\end{lemma}
We define on the space of kernels $\WW_{m,n}^\#$ a
natural representation of $SO(3)$, $\mathcal{U}$, which
is uniquely determined by
\begin{equation} \label{eq:rinvkerop}
H (\mathcal{U}(R) w_{m,n}) = \mathcal{U}(R) H(w_{m,n}) \mathcal{U}^*(R) , \quad \forall R \in SO(3) ,
\end{equation}
\cite{HH10-2}.
 The representation on $\WW_{m,n}^\#$ yields a natural representation
on $\WW_\xi^\#$, which is given by $(\mathcal{U}(R) w )_{m,n} = \mathcal{U}(R) w_{m,n}$ for all $R \in SO(3)$.
It lifts to a represention on $\WW_\xi$ by setting $(\mathcal{U}(R) w )(z) =  \mathcal{U}(R) w(z) $ for all $w \in \WW_\xi$.
We say that a kernel $w_{m,n} \in \WW_{m,n}^\# $ is rotation invariant if $\mathcal{U}(R) w_{m,n} = w_{m,n}$ and we say
a kernel $w \in \WW_{\xi}^\# $ is rotation invariant if $\mathcal{U}(R) w = w$.
We will use the following lemma which is proven in \cite{HH10-2}.

\begin{lemma} \label{lem:wHinvequiv} (i)
Let $w_{m,n} \in \WW_{m,n}^\#$. Then $H(w_{m,n})$ is rotation invariant if and only if $w_{m,n}$ is rotation invariant.
  Let $w \in \WW_{\xi}^\#$. Then $H(w)$ is rotation invariant if and only if $w$ is rotation invariant.
(ii) If
$w_{m,n} \in \WW_{m,n}^\#$ with $m+n=1$  is rotation invariant, then $w_{m,n} = 0$.
\end{lemma}

We will use the following polydiscs to define the renormalization transformation.
\begin{align*}
& \mathcal{B}^\#(\alpha,\beta,\gamma) := \left\{ w \in \mathcal{W}_\xi^\# \left| \| \partial_r w_{0,0} - 1 \|_\infty \leq \alpha , \
|w_{0,0}(0) | \leq \beta
, \ \| w_{\geq 1} \|_{\xi}^\# \leq \gamma \right. \right\}  ,
\\
& \mathcal{B}(\alpha,\beta,\gamma)   := \left\{ w \in \mathcal{W}_\xi^{} \left| \sup_{z \in D_{1/2}} \| \partial_r w_{0,0}(z) - 1 \|_\infty \leq \alpha , \,
\sup_{z \in D_{1/2}} | w_{0,0}(z,0) + z | \leq \beta
, \, \| w_{\geq 1} \|_{\xi} \leq \gamma \right. \right\} \\
& \mathcal{B}_0(\alpha,\beta,\gamma) := \{ w \in \mathcal{B}(\alpha,\beta,\gamma) |
 w(z) \ {\rm is \ rotation \ invariant \ for \ all } \ z \in D_{1/2} \  \}
\end{align*}
To control the derivatives with respect to $\beta$, we introduce the following extended polydisc.
\begin{align*}
&\mathcal{B}^{(\#, k)}(\alpha,\beta,\gamma)
:= \Big\{ w \in \mathcal{W}_\xi^{(\#,k)} \Big|  \| \partial_r w_{0,0}  - 1 \|_{C^{k}(\R ; C_B[0,1])} \leq \alpha ,
\|w_{0,0}(0) \|_{C^{k}(\R)} \leq \beta
, \ \| w_{\geq 1} \|_{\xi}^{(\#,k)} \leq \gamma    \Big\}  \\
&\mathcal{B}^{(k)}(\alpha,\beta,\gamma)   := \Bigg\{ w \in \mathcal{W}_\xi^{(k)} \Bigg| \sup_{z \in D_{1/2}}
\| \partial_r w_{0,0}(z)  - 1 \|_{C^{k}(\R; C_B[0,1])} \leq \alpha , \\
&\qquad \qquad \qquad  \sup_{z \in D_{1/2}} \| w_{0,0}(z,0) + z \|_{C^{k}(\R)}  \leq \beta
, \ \| w_{\geq 1} \|_{\xi}^{( k)} \leq \gamma   \Bigg\}  \\
& \mathcal{B}_0^{(k)}(\alpha,\beta,\gamma)  :=
\{ w \in \mathcal{B}^{(k)}(\alpha,\beta,\gamma)  | w(z)  \ {\rm is \ rotation \  invariant \ for \ all \ }  z \in D_{1/2}  \ \} .
\end{align*}

\section{Initial Feshbach Transformation}

\label{sec:ini}

In this section we shall assume that the assumptions of  Hypothesis (H) hold. Without loss of generality, see Section \ref{sec:outline}, we
assume  that the distance between the lowest  eigenvalue of $H_{\rm at}$ and the rest of the spectrum is one, that is
\begin{equation} 
 \inf \left(  \sigma( H_{\rm at}) \setminus \{ E_{\rm at} \}  \right) - E_{\rm at} = 1  .
\end{equation}
Let $\chi_1$ and  $\chib_1$ be two functions  in $C^\infty(\R_+;[0,1])$  with $\chi_1^2 + \chib_1^2 = 1$,
$\chi_1 = 1$ on $[0,3/4)$, and $\supp \chi_1 \subset [0,1]$.
We use the abbreviations $\chi_1 = \chi_1(H_f)$ and $\chib_1  = \chib_1(H_f)$.
It should be clear from the context whether $\chi_1$ or $\chib_1$ denotes a function or an operator.
By  $\varphi_{\rm at}$ we denote  a fix choice for a  normalized eigenstate of $H_{\rm at}$ with eigenvalue
 $E_{\rm at}$ and by  $P_{\rm at}$ we denote
the eigen-projection of $H_{\rm at}$ corresponding to the
 eigenvalue $E_{\rm at}$.
 By Hypothesis (H) the range of
$P_{\rm at}$ is one dimensional. Thus
to every $\psi \in \ran P_{\rm at} \otimes P_{\rm red}$
there exists a unique $\iota(\psi) \in \HH_{\rm red}$ such that
$\psi = \varphi_{\rm at} \otimes \iota(\psi)$. It follows that
 $\iota : \ran P_{\rm at} \otimes P_{\rm red} \to  \HH_{\rm red}$ is unitary
and commutes with the $SO(3)$ action.
We will use $\iota$ to identify the range of
$P_{\rm at} \otimes P_{\rm red}$ with $\HH_{\rm red}$.
We define  $\chi^{(I)}(r) := P_{\rm at} \otimes \chi_1(r)$ and
$\chib^{(I)}(r)  = \bar{P}_{\rm at}  \otimes 1 + P_{\rm at} \otimes \chib_1(r)$,
with $\bar{P}_{\rm at} = 1 - P_{\rm at}$. We set
$\chi^{(I)} := \chi^{(I)}(H_f)$ and  $\chib^{(I)} := \chib^{(I)}(H_f)$. It follows directly from the definition that
 ${\chi^{(I)}}^2 + {\chib^{(I)}}^2 = 1$. We use an initial transformation based on the smooth Feshbach map
and its associated auxiliary operator, see  Appendix B.

\begin{theorem} \label{thm:inimain1}  Assume Hypothesis (H).  Let $k \in \N$.
For any $0 < \xi < 1$ and any positive numbers $\delta_1,\delta_2,\delta_3$ there exists a positive number
$g_0$ such that following is satisfied. For all
$(g,\beta,z) \in D_{g_0} \times \R  \times D_{1/2}$
the pair of operators
$(H_{g,\beta} - z - E_{\rm at}, H_0 -  z - E_{\rm at} )$ is a Feshbach pair for $\chi^{(I)}$.
The operator valued map
$$
Q_{\chi^{(I)}}(g,\beta,z) := Q_{\chi^{(I)}} (H_{g,\beta} - z - E_{\rm at}, H_0 -  z - E_{\rm at} )
$$
is uniformly bounded in $(g,\beta,z)$ and the function $(g,z) \mapsto  Q_{\chi^{(I)}}(g,\cdot,z)$ is in $
C^{\omega}_B( D_{g_0} \times D_{1/2} ;  C_B^k(\R ; \mathcal{B}(\HH_{\rm red} , \HH ))$.
There exists a unique kernel
$w^{(0)}(g,\beta,z) \in \WW_\xi^{\#}$ such that
\begin{equation} \label{eq:inimainA}
H(w^{(0)}(g,\beta,z)) =  \iota ( F_{\chi^{(I)}}( H_{g,\beta} - z - E_{\rm at}  , H_0 - z - E_{\rm at} ) \upharpoonright \ran P_{\rm at} \otimes P_{\rm red} ) \iota^{-1}.
\end{equation}
Moreover,  $w^{(0)}$ satisfies the following properties.
\begin{itemize}
\item[(a)]  We have
$w^{(0)}(g) := w^{(0)}(g,\cdot , \cdot ) \in \mathcal{B}_0^{(k)}(\delta_1,\delta_2,\delta_3)$
for all $g \in D_{g_0}$.
\item[(b)]
$w^{(0)}(g,\beta,\cdot)$ is a symmetric kernel  for all $(g,\beta) \in ( D_{g_0}  \cap \R) \times \R $.
\item[(c)]
The function $(g, z,\beta ) \mapsto w^{(0)}(g,\beta,z)$ is in
$C^{\omega,k}_B(      D_{g_0} \times D_{1/2} \times \R ; \WW_\xi^{\#})$.
\end{itemize}
\end{theorem}
The remaining part of this section is devoted to the proof of  Theorem  \ref{thm:inimain1}.
Throughout this section we assume that
\begin{equation} \label{eq:zzetarelation}
z = \zeta -  E_{\rm at }  \in D_{1/2} .
\end{equation}
To prove Theorem \eqref{thm:inimain1},    we write the interaction part of the Hamiltonian in terms
of  integral kernels as follows,
$$
H_{g,\beta}  =  H_{\rm at} +   H_f + : W_{g, \beta} :  ,
$$
\begin{equation} \label{eq:sumofws}
W_{g,\beta} := \sum_{m+n=1,2}  W_{m,n}(g,\beta) .
\end{equation}
where $W_{m,n}(g,\beta) := \underline{H}_{m,n}(w_{m,n}^{(I)}(g,\beta))$
with
\begin{align}
 \label{eq:defhlinemn}
& \underline{H}_{m,n}(w_{m,n}) := \int_{{(\underline{\R}^3)}^{m+n}} \frac{ dK^{(m,n)}}{|K^{(m,n)}|^{1/2}}
a^*(K^{(m)}) w_{m,n}(K^{(m,n)}) a(\widetilde{K}^{(n)}) ,
\end{align}
and
\begin{align}
&w^{(I)}_{1,0}(g,\beta)( K) :=  2  g  \sum_{j=1}^N p_j \cdot  \varepsilon(k,\lambda)  \frac{\kappa_{\Lambda}(k)e^{ i \beta k \cdot x_j }}{\sqrt{2}} ,
\label{defofwI}  \\
&w^{(I)}_{1,1}(g,\beta)(K,\widetilde{K}) :=  g^2 \sum_{j=1}^N \varepsilon(k,\lambda) \cdot \varepsilon(\widetilde{k},\widetilde{\lambda}) \frac{\kappa_{\Lambda}(k)e^{- i \beta k \cdot x_j }}{\sqrt{2}}   \frac{\kappa_{\Lambda}(\widetilde{k})e^{ i \beta \widetilde{k}\cdot  x_j }}{\sqrt{2}} ,
\nonumber \\
&w^{(I)}_{2,0}(g,\beta)( K_1, K_2 ) :=   g^2 \sum_{j=1}^N  \varepsilon(k_1,\lambda_1) \cdot \varepsilon({k}_2,{\lambda}_2)  \frac{\kappa_{\Lambda}(k_1)e^{ - i \beta k_1 \cdot x_j }}{\sqrt{2}}   \frac{\kappa_{\Lambda}(k_2)e^{ - i \beta k_2 \cdot x_j }}{\sqrt{2}} ,
\nonumber
\end{align}
$w^{(I)}_{0,1}(g,\beta)(\widetilde{K}) := w^{(I)}_{0,1}(\overline{g},\beta)(\widetilde{K})^*$,  and
$w^{(I)}_{0,2}(g,\beta)(\widetilde{K}_1,\widetilde{K}_2) := \overline{w^{(I)}_{2,0}(\overline{g},\beta)(\widetilde{K}_1,\widetilde{K}_2  )}$.
We note that  \eqref{eq:defhlinemn} is understood in the sense of forms, c.f.  Appendix A.
We set
\begin{align*}
w^{(I)}_{0,0}(z)(r)  := H_{\rm at} - z + r   .
\end{align*}
By $w^{(I)}$ we denote the vector consisting of the components $w^{(I)}_{m,n}$ with  $m+n=0,1,2$.

The next theorem establishes the Feshbach property.
To state it, we denote by   ${P}_0$ the orthogonal projection onto the closure of ${\ran \chib^{(I)}}$.
We will use the  convention that
$( H_0 - z)^{-1} \chib^{(I)}$
stands for $( H_0 - z \upharpoonright \ran \chib^{(I)} ) )^{-1} \chib^{(I)}$, and that
$( H_0 - z)^{-1} P_0$ stands for $( H_0 - z \upharpoonright \ran {P}_0  )^{-1} P_0 $.
The proof of the Feshbach property  is based on the fact that
\begin{equation} \label{ini:eq1}
{\rm inf} \sigma ( H_0 \upharpoonright \ran {P_0} ) \geq E_{\rm at} + \sfrac{3}{4} ,
\end{equation}
which follows directly from the definition, and the
fact that the interaction part of the Hamiltonian is bounded with respect to the free Hamiltonian.
A proof can be found in \cite{HH10-2}
\begin{theorem} \label{ini:thm1} Let $| E_{\rm at} - \zeta  | < \frac{1}{2}$. Then
\begin{equation}  \label{ini:thm1:eq1}
\left\| ( ( H_0 - \zeta ) \upharpoonright  \ran P_0  )^{-1} \right\| \leq 4
\end{equation}
There is a $C<\infty$ and $g_0 > 0$ such that for all $\beta$ and $|g| < g_0$,
\begin{equation} \label{ini:thm1:eq2}
\left\| ( H_0 - \zeta)^{-1} \chib^{(I)} W_{g,\beta} \right\| \leq C |g| , \quad \left\| W_{g, \beta} ( H_0 - \zeta)^{-1} \chib^{(I)}  \right\| \leq C |g| ,
\end{equation}
and $(H_{g,\beta}-\zeta,H_0-\zeta)$ is a Feshbach pair for $\chi^{(I)}$.
\end{theorem}

\begin{theorem} \label{ini:thm1b} For $g_0$ sufficiently small
\begin{equation}
(g,z) \mapsto  Q_{\chi^{(I)}}(g,\cdot,z)
\end{equation}
is in  $C^{\omega}_B(D_{g_0} \times D_{1/2} ; C_B^k(\R ; \mathcal{B}(\HH_{\rm red} , \HH )))$.
\end{theorem}

We write
$$\langle x \rangle := \left( 1 + \sum_{j=1}^N | x_j|^2 \right)^{1/2} . $$
We will use the Leibniz rule for higher derivatives
\begin{equation}\label{eq:higherleibniz}
\partial_\beta^l  ( f_1 \cdots f_L ) =
\sum_{\unn \in \N_0^L : |\unn|  = l} \frac{l!}{\unn!}
f_1^{(n_1)} \cdots f_L^{(n_L)} .
\end{equation}
\begin{proof} For notational simplicity we set  $W = W_{g,\beta}$ and $Q_{\chi^{(I)}} =  Q_{\chi^{(I)}}(g,\beta,z)$. We have
\begin{align} \label{eq:jun14-1}
 Q_{\chi^{(I)}} &=  \chi^{(I)}   - \sum_{n=0}^\infty (-1)^n \chib^{(I)} \left( (H_0 - \zeta )^{-1} \chib^{(I)} W \chib^{(I)} \right)^n (H_0  - \zeta )^{-1} \chib^{(I)} W \chi^{(I)} .\nonumber \\
& = \chi^{(I)}   - \sum_{n=0}^\infty (-1)^n \left(\chib^{(I)}(H_0 - \zeta )^{-1} \chib^{(I)} W \right)^{n+1}\chi^{(I)}.
\end{align}
Formally differentiating $l$ times with respect to $\beta$, the $n^{\rm{th}}$ term under the summation sign generates  $(n+1)^{l}$ terms, each of the form
\begin{equation} \label{eq:jun14-2}
\left(\chib^{(I)} (H_0 - \zeta)^{-1} \chib^{(I)} \partial_\beta^{l_{n+1}} W \right) \cdots
\left(\chib^{(I)} (H_0 - \zeta)^{-1} \chib^{(I)} \partial_\beta^{l_{1}} W  \right)\chi^{(I)}  ,
\end{equation}
where $l_1+ \cdots + l_{n+1} = l$. We write
\begin{equation}  \label{eq:jun14-3}
\chib^{(I)} (H_0 - \zeta )^{-1} \chib^{(I)} = \left( \chib^{(I)} \right)^2 (H_0 + 2 - \zeta )^{-1} + 2 (H_0 - \zeta)^{-1} \left( \chib^{(I)}\right)^2 (H_0 + 2 - \zeta )^{-1} .
\end{equation}
It is well known that $\| e^{ \gamma_1 \langle x \rangle } P_{\rm at} \| < \infty$ for some $\gamma_1 > 0$ \cite{reesim4}. Define
\begin{equation} \nonumber
\gamma_{j+1} = \left(1- k^{-1}\sum_{t=1}^j ( 1 - \delta_{l_t,0} )\right) \gamma_1 ; \\ j=1, \cdots , n .
\end{equation}
Since $\sum_{j=1}^n ( 1 - \delta_{l_j,0} ) \leq k$ , $\gamma_{n+1} \geq 0$. With
\begin{equation} \label{basicderexpressionj5}
G_j = \left( e^{\gamma_{j+1} \langle x \rangle }  \chib^{(I)} (H_0 - \zeta )^{-1} \chib^{(I)} e^{- \gamma_{j+1} \langle x \rangle } \right)
\left( e^{\gamma_{j+1} \langle x \rangle } \partial_\beta^{l_j} W e^{- \gamma_j \langle x \rangle }  \right)
\end{equation}
the expression in \eqref{eq:jun14-2} can be written as
\begin{equation}  \label{eq:jun14-4}
 e^{- \gamma_{n+1} \langle x \rangle } \left( G_{n+1} \cdots G_1 \right) e^{\gamma_1 \langle x \rangle } \chi^{(I)} .
\end{equation}
We claim that for small enough $\gamma_1  > 0$ (chosen independent of $n$), for $| g | \leq 1$, and for $\zeta \in D_{1/2} + E_{\rm at}$
\begin{equation}  \label{eq:jun14-5}
\| G_j \| \leq C | g | ,
\end{equation}
where $C$ is independent of $j$, $\zeta$, $\beta$, and $n$. It is clear that
$$
\left\| (H_0 + i )^{-1} e^{\gamma_{j+1} \langle x \rangle } \partial_\beta^{l_j} W e^{- \gamma_j \langle x \rangle } \right\| \leq C_1 |g| ,
$$
since if $l_j >0$, $\gamma_j - \gamma_{j+1} = \gamma_1 /k $. (There is a slight subtlety here with the term $W_{0,1}^{(I)}$ which contains $(e^{i\beta k \cdot x_j})(p_j \cdot \epsilon(k,\lambda))$. But note that the two terms in parentheses commute so that the bound is indeed independent of $\beta$.) It remains to show
$$
e^{\gamma \langle x \rangle } \chib^{(I)} (H_0 - \zeta )^{-1} \chib^{(I)} e^{- \gamma \langle x \rangle } (H_0 + i )
$$
is bounded with bound independent of $\gamma$ for small $\gamma$  and $\zeta \in D_{1/2} + E_{\rm at}$ .
We have
\begin{align*}
&H_0(\gamma) := e^{\gamma \langle x \rangle } H_0 e^{- \gamma  \langle x \rangle } = H_{\rm at}(\gamma) + H_f \\
 & H_{\rm at}(\gamma) := H_{\rm at} + i \gamma   \left( \frac{x}{\langle x \rangle} \cdot p + p \cdot \frac{x}{\langle x \rangle } \right)
 - \gamma^2 \frac{|x|^2}{\langle x \rangle^2} .
\end{align*}
and thus for all small $\gamma$
$$
\| ( H_0(\gamma) + 2 - \zeta )^{-1} (H_0 + i ) \| \leq C_2 .
$$
For $\zeta \in D_{1/2} + E_{\rm at}$. Clearly
$\| e^{\gamma \langle x \rangle } ({ \chib^{(I)} })^2 e^{- \gamma \langle x \rangle}  \| \leq c_3$ for
$\gamma$ small so from \eqref{eq:jun14-3} it remains to bound
$$
e^{\gamma \langle x \rangle } (H_0 - \zeta )^{-1} ({\chib^{(I)}})^2 e^{- \gamma \langle x \rangle } .
$$
Since $({\chib^{(I)}})^2 = P_{\rm at} \otimes \chib_1(H_f)^2 +  \bar{P}_{\rm at}  \otimes 1 $ and
$$
(H_0 - \zeta )^{-1} P_{\rm at} \otimes  1 = ( 1 \otimes H_f + E_{\rm at} - \zeta )^{-1} P_{\rm at} \otimes 1
$$
we must only control
$$
e^{\gamma \langle x \rangle } (H_0 - \zeta )^{-1} ( \bar{P}_{\rm at}  \otimes 1  ) e^{- \gamma \langle  x \rangle }  .
$$
We write
\begin{equation}
(H_{\rm at} + t - \zeta )^{-1} \bar{P}_{\rm at}  = \frac{1}{2\pi i} \int_{\Gamma} ( w + t - \zeta )^{-1} ( w - H_{\rm at} )^{-1} dw ,
\end{equation}
where $\Gamma$ is the contour $\Gamma_- - \Gamma_+$ with
$$
\Gamma_{\pm}(s) = E_{\rm at} + 3/4 + e^{\pm i \pi /4 } s , \quad 0 \leq s < \infty .
$$
Thus (using an analytic continuation argument)
\begin{eqnarray} 
\lefteqn{ e^{\gamma \langle x \rangle } (H_0 - \zeta )^{-1} (  \bar{P}_{\rm at} \otimes 1 ) e^{- \gamma \langle x \rangle } } \nonumber  \\
&& = \frac{1}{2 \pi i } \int_\Gamma ( w - H_{\rm at}(\gamma) )^{-1} \otimes ( w + H_f - \zeta )^{-1} dw .  \label{eq:june14-6}
\end{eqnarray}
The expression  \eqref{eq:june14-6} is bounded using a numerical range argument for large $w$ and a perturbation argument for
small $w$. These estimates require $\gamma$ to be small. We have thus shown \eqref{eq:jun14-5}.
Moreover, it follows from the estimates above and Taylor's theorem with remainder that the derivative with respect to $\beta$ in \eqref{basicderexpressionj5} and thus \eqref{eq:jun14-2}
exists with respect to the operator norm topology.
It follows that for $l \leq k$,
\begin{equation} \label{eq:boundonQinitialder}
\left\| \partial_\beta^l \left( Q_{\chi^{(I)}} - \chi^{(I)} \right) \right\| \leq \sum_{n=0}^\infty (n+1)^{l} ( C |g| )^{n+1}
\left\| e^{\gamma_1 \langle x \rangle } P_{\rm at} \right\|
\end{equation}
for all $\beta \in \R$ and $\zeta \in D_{1/2} + E_{\rm at}$.
If   $g_0 >0 $ is sufficiently small, then \eqref{eq:boundonQinitialder} converges
for $|g| < g_0$.
The  expression in  \eqref{eq:jun14-2}     is complex differentiable in
$\zeta  \in D_{1/2} + E_{\rm at}$ and in $g$
with respect to the operator norm topology.
The bounds
 \eqref{eq:jun14-5} and   \eqref{eq:boundonQinitialder}  imply uniform convergence and that
$(g,z) \mapsto  Q_{\chi^{(I)}}(g,\cdot,z)$
is in  $C^{\omega}_B(D_{g_0} \times D_{1/2} ; C_B^k(\R ; \mathcal{B}(\HH_{\rm red} , \HH )))$.
\end{proof}

Next we want to show that there exists a $w^{(0)}(g,\beta,z) \in \WW_\xi^\#$ such that  \eqref{eq:inimainA} holds.
Uniqueness will follow from Theorem \ref{thm:injective}.
In view of Theorem  \ref{ini:thm1} for $z = \zeta -  E_{\rm at}  \in D_{1/2}$  and $g$ sufficiently small
we can define the Feshbach map and express it in terms of a Neumann series.
\begin{eqnarray*}
\lefteqn{ F_{\chi^{(I)}}( H_{g,\beta} - \zeta, H_0 - \zeta)  \upharpoonright X_{\rm at} \otimes  \HH_{\rm red}  } \\
&& = \left( T  + \chi^{}  W \chi^{} -    \chi^{}  W  \chib^{}        (  T +  \chib^{}  W_{}\chib^{}   )^{-1}
\chib^{}   W_{}  \chi^{} \right) \upharpoonright X_{\rm at} \otimes  \HH_{\rm red} \\
&& = \left( T^{} + \chi W^{} \chi - \chi W^{} \chib \sum_{n=0}^\infty \left( - {T^{}}^{-1} \chib W^{} \chib \right)^n
{T^{}}^{-1} \chib W^{} \chi \right) \upharpoonright  X_{\rm at} \otimes  \HH_{\rm red} \; ,
\end{eqnarray*}
where here we used   the abbreviations
$T^{} = H_0 - \zeta$, $W^{} = W^{}_{g,\beta}$, $\chi = \chi^{(I)}$,  $\chib = \chib^{(I)}$ and $X_{\rm at} = \ran P_{\rm at}$.
We normal order above expression, using the pull-through formula. To this end we use a generalized version
of the Wick theorem, see \cite{BFS99}, see also \cite{HH10-2}  Appendix B.
Moreover we will use the definition
\begin{eqnarray} \label{eq:defofWW}
\underline{W}_{p,q}^{m,n}[w](K^{(m,n)}) &:=& \int_{{(\underline{\R}^3)}^{p+q}} \frac{d X^{(p,q)}}{|X^{(p,q)}|^{1/2}} a^*(X^{(p)}) w_{m+p,n+q}(K^{(m)}, X^{(p)}, \widetilde{K}^{(n)},
 \widetilde{X}^{(q)}) a(\widetilde{X}^{(q)})  . \nonumber
\end{eqnarray}
We obtain a sequence of integral kernels
$\widetilde{w}^{(0)}$, which  are given as follows.
For $M+N \geq 1$,
\begin{eqnarray} \label{eq:defofwmnschlange}
\lefteqn{ \widetilde{w}^{(0)}_{M,N}(g,\beta,z)(r , K^{(M,N)})}   \label{initial:eq7}  \\ &&= ( 8 \pi )^{\frac{M+N}{2}} \sum_{L=1}^\infty (-1)^{L+1}
\sum_{\substack{ (\umm,\upp,\unn,\uqq)  \in \N_0^{4L}: \\ |\umm|=M,  |\unn|=N, \\  1 \leq m_l+p_l+q_l+n_l \leq  2 } }   \prod_{l=1}^L \left\{
\binom{ m_l + p_l}{ p_l}  \binom{ n_l + q_l}{   q_l }
\right\}  \nonumber \\ && \times
V_{(\umm,\upp,\unn,\uqq)}[w^{I}(g,\beta,\zeta)](r,K^{(M,N)}).  \nonumber
\end{eqnarray}
 Furthermore,
\begin{align*}
\widetilde{w}^{(0)}_{0,0}(g,\beta,z)(r) = - z + r  + \sum_{L=2}^\infty (-1)^{L+1}
\sum_{(\upp,\uqq)\in \N_0^{2L}: p_l+q_l = 1, 2}
V_{(\uzz,\upp,\uzz,\uqq)}[w^{(I)}(g,\beta,\zeta)](r) \; .
\end{align*}
Above we have used  the definition
\begin{eqnarray} \label{eq:defofV}
\lefteqn{ V_{\umm,\upp,\unn,\uqq}[w](r, K^{(|\umm|,|\unn|)}) := } \\
&&
\left\langle \varphi_{\rm at} \otimes \Omega  ,  F_0[w](H_f + r) \prod_{l=1}^L \left\{
\underline{W}_{p_l,q_l}^{m_l,n_l}[w](K^{(m_l,n_l)})
 F_l[w](H_f + r  + \widetilde{r}_l   ) \right\} \varphi_{\rm at} \otimes \Omega
\right\rangle \nonumber ,
\end{eqnarray}
where for $l=0,L$ we set  $F_l[w](r) := \chi_1(r )$ , and for $l=1,...,L-1$ we set
\begin{eqnarray*}
F_l[w](r) := F[w](r) := \frac{ { \chib}^{(I)}(r  )^2}{ w_{0,0}(r )}  .
\end{eqnarray*}
Moreover, we used the notation
\begin{align}
& r_l := \Sigma[\widetilde{K}_1^{(n_1)}] + \cdots + \Sigma[\widetilde{K}_{l-1}^{(n_{l-1})}] + \Sigma[{K}_{l+1}^{(m_{l+1})}] + \cdots + \Sigma[{K}_L^{(m_L)}] , \label{eq:rldef}
\\
& \widetilde{r}_l := \Sigma[\widetilde{K}_1^{(n_1)}] + \cdots + \Sigma[\widetilde{K}_{l}^{(n_{l})}] + \Sigma[{K}_{l+1}^{(m_{l+1})}] + \cdots + \Sigma[{K}_L^{(m_L)}] . \label{eq:rltildedef}
\end{align}
We have
$w^{(0)}(g,\beta,z) = \left( \widetilde{w}^{(0)} \right)^{({\rm sym})}(g,\beta,z) $.
So far we have determined $w^{(0)}$ on a formal
level only.

\begin{lemma}   \label{initial:thmE22} Let $k \in \N_0$.
The function $(g, \zeta ,\beta ) \mapsto V_{\umm,\upp,\unn,\uqq}[w^{(I)}(g,\beta,\zeta)]$ is in
$C^{\omega,k}_B(     \C \times D_{1/2}(E_{\rm at}) \times \R ; \WW_{|\umm|,|\unn|}^{\#})$.
There exists a finite constant $C$ such that for all
$(g,\beta, \zeta) \in \mathbb{C} \times \R \times  D_{1/2}(E_{\rm at})$ we have
\begin{equation} 
\max_{0 \leq l \leq k} \| \partial_\beta^l  V_{\umm,\upp,\unn,\uqq}[w^{(I)}(g,\beta,\zeta)] \|^{\#}  \leq L^{k+1} C^L  | g |^{|\umm|+|\unn| + |\upp| + | \uqq|} .
\end{equation}
\end{lemma}

\begin{proof} For compactness we shall drop the $\zeta$ and $\beta$ dependence in the notation.
We show
\begin{equation} \label{initial:thmE22:eq1A}
| \partial_\beta^s V_{\umm,\upp,\unn,\uqq}[w^{(I)}(g)](r,K^{(|\umm|,|\unn|)} ) | \leq L^s C^L  | g |^{|\umm|+|\unn| + |\upp| + | \uqq|} .
\end{equation}
\begin{equation} \label{initial:thmE22:eq2}
\left| \partial_r \partial_\beta^s V_{\umm,\upp,\unn,\uqq}[w^{(I)}(g)](r,K^{|\umm|,|\unn|})\right| \leq L^{s+1}C^L  | g|^{|\umm| + |\unn| + |\upp| + |\uqq|} .
\end{equation}
Consider
\begin{eqnarray}
\lefteqn{
\partial_\beta^s V_{\umm,\upp,\unn,\uqq}[w^{(I)}](r,K^{(|\umm|,|\unn|)}) } \label{eq:derVbetaj5-1} \\
&& = \sum_{\substack{\ujj \in \N_0^L \\ |\ujj|=s}} \frac{s!}{\ujj!} \Bigg\langle \varphi_{\rm at} \otimes \Omega, F_0[w^{(I)}](H_f + r)\nonumber \\&&
\times \prod_{l=1}^L \left\{ \partial_\beta^{j_l} \underline{W}^{m_l,n_l}_{p_l,q_l}[w^{(I)}](K^{(m_l,n_l)}) F_l[w^{(I)}](H_f + r + \widetilde{r}_l ) \right\} \varphi_{\rm at} \otimes\Omega \Bigg\rangle .
\nonumber
\end{eqnarray}

To  estimate (\ref{eq:derVbetaj5-1}) we will use the same technique as in Theorem \ref{ini:thm1b}.  For $l = 1,\cdots,L-1$ define
\begin{equation} \label{eq:derVbetaj5-2}
A_l^{j_l} = e^{\gamma_{l+1} \langle x \rangle } \partial_\beta^{j_l} \underline{W}_{p_l,q_l}^{m_l,n_l}[w^{(I)}](K^{(m_l,n_l)})F_l[w^{(I)}](H_f + r + \widetilde{r}_l) e^{-\gamma_{l} \langle x \rangle } .
\end{equation}
and similarly for $A_L^{j_L}$ except that we replace $F_L[w^{(I)}](H_f + r + \widetilde{r}_l)$ by $(H_0 - E_{\rm at} + 1 )^{-1}$.  Here
$$ \gamma_{l+1} = \left(1- k^{-1}\sum_{t=1}^l ( 1 - \delta_{j_t,0} )\right) \gamma_1 ; \quad l=1, \cdots , L .
$$
Note again that since $s\leq k, \gamma_{L+1} \geq 0$.  It follows that
$$
| \partial_\beta^s V_{\umm,\upp,\unn,\uqq}[w^{(I)}](r,K^{(|\umm|,|\unn|)})| \leq
\sum_{ \substack{\ujj \in \N_0^L \\ |\ujj|=s}       }
\frac{s!}{\ujj!} \left( \prod_{l=1}^L \| A_l^{j_l} \| \right) \| e^{\gamma_1 \langle x \rangle} \varphi_{\rm at} \|.
$$
We will show the bound
\begin{equation} \label{eq:Abd}
\| A_l^{j_l} \| \leq C | g |^{m_l+p_l+n_l+q_l},
\end{equation}
which gives \eqref{initial:thmE22:eq1A}.
We write for $l \leq L$,
\begin{eqnarray}\label{estofAia}
 A_l^{j_l} && = e^{\gamma_{l+1} \langle x \rangle}\partial_\beta^{j_l} \underline{W}_{p_l,q_l}^{m_l,n_l}[w^{(I)}](K^{(m_l,n_l)})(H_0 - E_{\rm at} + 1 )^{-1}
e^{- \gamma_{l} \langle x \rangle } \\ \nonumber
&&
\times
e^{ \gamma_{l} \langle x \rangle } (H_0 - E_{\rm at} + 1 ) F_l[w^{(I)}](H_f + r + \widetilde{r}_l)
 e^{-\gamma_{l} \langle x \rangle } .
\end{eqnarray}
First we estimate the second factor. To this end we write
$$
F_l[w^{(I)}](H_f + r + \widetilde{r}_l) = (H_0 - E_{\rm at} - z + r + \widetilde{r}_l)^{-1}
\left( \overline{P}_{\rm at} \otimes 1 + P_{\rm at} \otimes \chi_1^2(H_f + r + \widetilde{r}_l )
\right) .
$$
Since $e^{\gamma_{l} \langle x \rangle } P_{\rm at} e^{- \gamma_{l} \langle x \rangle }$ is bounded for $\gamma_1$ small, it
is clear that
\begin{eqnarray*}
\lefteqn{
\| e^{\gamma_{l}\langle x \rangle } (H_0 - E_{\rm at} + 1 ) F_l[w^{(I)}](H_f + r + \widetilde{r}_l) e^{-\gamma_{l} \langle x \rangle } \| } \\
&&
\leq C_1 + (r + \widetilde{r}_l) \| e^{\gamma_{l} \langle x \rangle }(H_0 - E_{\rm at} - z + r + \widetilde{r}_l )^{-1}
\overline{P}_{\rm at} \otimes 1 e^{-\gamma_{l} \langle x \rangle } \| .
\end{eqnarray*}
For $u \geq 0$ we write
\begin{equation}
(H_{\rm at} - E_{\rm at} - z + u )^{-1} \overline{P}_{\rm at} = \frac{1}{2 \pi i} \int_\Gamma ( w - E_{\rm at} - z + u )^{-1} (w - H_{\rm at})^{-1} dw
\end{equation}
where $\Gamma$ is the contour $\Gamma = \Gamma_- - \Gamma_+$ with
$$
\Gamma_{\pm}(t) = E_{\rm at} + 3/4 + e^{\pm i \frac{\pi}{4} } t  , \quad 0 \leq t < \infty .
$$
and obtain for $\lambda$ small
\begin{equation} \label{eq:i4}
e^{ \lambda  \langle x \rangle} ( H_{\rm at} - E_{\rm at} - z + u )^{-1} \overline{P}_{\rm at} e^{ - \lambda \langle x \rangle } =
\frac{1}{2 \pi i} \int_\Gamma ( w - E_{\rm at} - z + u )^{-1} ( w - H_{\rm at}(\lambda) )^{-1} dw ,
\end{equation}
where $H_{\rm at}(\lambda)$ is given as in the proof of Theorem \ref{ini:thm1b}.  As in that proof we estimate (\ref{eq:i4})
for large $w \in \Gamma$ using a numerical range estimate to bound $\|( w - H_{\rm at}(\lambda) )^{-1}\|$ while for small
$w \in \Gamma$ the resolvent can be bounded using
$$
\|(H_{\rm at}(\lambda) - H_{\rm at} ) (-\Delta + 1 )^{-1} \| = O(|\lambda|)
$$
for small $\lambda$.
Then using the spectral
theorem, which allows us to substitute $u= H_f + r + \widetilde{r}_l$, we obtain for small $\gamma_1$
$$
y \| e^{\gamma_{l} \langle x \rangle } ( H_0 - E_{\rm at} - z + y )^{-1} \overline{P}_{\rm at} \otimes 1 e^{-\gamma_{l} \langle x \rangle } \| \leq C
$$
independent of $y \geq 0$.  In order to show \eqref{eq:Abd} for $1 \leq l \leq L$ it remains to bound  the first factor on the right hand side of
(\ref{estofAia}).
Using $\| ( \frac{x}{\langle x \rangle } \cdot p + p \cdot \frac{x}{\langle x \rangle } ) (-\Delta + 1 )^{-1} \| < \infty$ and
Hypothesis (H) we see that
$$
\| ( -\Delta \otimes 1 + 1 \otimes H_f + 1)e^{\gamma_{l} \langle x \rangle } (H_0 - E_{\rm at} + 1 )^{-1} e^{-\gamma_{l} \langle x \rangle } \|
$$
is bounded uniformly in $L$ for small $\gamma_1$. Thus to prove \eqref{eq:Abd} we need only bound
$$
e^{\gamma_{l+1} \langle x \rangle } \partial_\beta^{j_l}
\underline{W}_{p_l,p_l}^{m_l,n_l}[w^{(I)}](K^{(m_l,n_l)}) e^{-\gamma_{l} \langle x \rangle }
(-\Delta \otimes 1  + 1 \otimes H_f + 1 )^{-1}
$$
or carrying out the differentiations with respect to $\beta$ (if any) we need to bound
$$
\underline{W}^{m_l,p_l}_{p_l,q_l}(e^{\gamma_{l+1} \langle  x \rangle }
{w}^{(I,j_l)} e^{-\gamma_{l} \langle x \rangle } )(K^{(m_l,n_l)})
(-\Delta \otimes 1 + 1 \otimes H_f + 1 )^{-1} ,
$$
where  ${w}^{(I,j_l)} := \partial_\beta^{j_l} w^{(I)}$.
Referring to \eqref{defofwI} we have
\begin{equation} \label{eq:xy1}
\| e^{\gamma_{l+1} \langle x \rangle } {w}_{1,0}^{(I,j_l)}(K)
e^{-\gamma_{l} \langle x \rangle }(-\Delta + 1 )^{-1/2} \|_{\HH_{\rm at} \to \HH_{\rm at}} \leq
c_1|g| {\kappa_\Lambda(k)} ,
\end{equation}
and similarly for ${w}_{0,1}^{(I,j_l)}$, while for $m+n=2$
\begin{equation} \label{eq:xy2}
\| e^{\gamma_{l+1} \langle x \rangle } {w}_{m,n}^{(I,j_l)}(K^{(m,n)}) e^{-\gamma_{l} \langle x \rangle }\|_{\HH_{\rm at} \to \HH_{\rm at}} \leq
c_2|g|^2   {\kappa_\Lambda(K^{(m)})}  {\kappa_\Lambda(\widetilde{K}^{(n)})}          ,
\end{equation}
where $\kappa_\Lambda(K^{(m)}) = \prod_{j=1}^m \kappa_\Lambda(k_j)$.
Given \eqref{eq:xy1} and \eqref{eq:xy2} we need only consider $\underline{W}^{m_l,n_l}_{p_l,q_l}$ with $p_l + q_l \geq 1$.
>From Lemma \ref{kernelopestimate}, if $m_l + n_l \leq 1$, $p_l = 1$, $q_l = 0$.
\begin{eqnarray*}
\lefteqn{ \left\| \int \frac{dX}{|X|^{1/2}} a^*(X) e^{\gamma_{l+1} \langle x \rangle }
 {w}^{(I,j_l)}_{m_l+1,n_l}(K^{(m_l)},X,\widetilde{K}^{(n_l)})e^{-\gamma_{l}\langle  x \rangle}
(-\Delta + 1 )^{-1/2} \otimes (H_f + 1 )^{-1/2} \right\|^2 }
\\
&& \leq \int \frac{dX}{|X|^2} \sup_{r \geq 0} \left\|
e^{\gamma_{l+1}\langle x \rangle } {w}^{(I,j_l)}_{m_l+1,n_l}(K^{(m_l)},X ,\widetilde{K}^{(n_l)} )
e^{-\gamma_{l} \langle x \rangle }(- \Delta + 1 )^{-1/2} \right\|_{\HH_{\rm at} \to \HH_{\rm at}} \frac{r + |X|}{r + 1}
\\
 &&\leq c |g|^{m_l+n_l+1} ,
\end{eqnarray*}
and similarly if $p_l=0,q_l=1$. If $p_l=q_l=1$
\begin{eqnarray*}
\lefteqn{ \left\| \int \frac{dX^{(1,1)}}{|X^{(1,1)}|^{1/2}} a^*(X_1) e^{\gamma_{l+1} \langle x \rangle }
{w}^{(I,j_l)}_{1,1}(X_1,\widetilde{X}_2)e^{-\gamma_{l}\langle  x \rangle}  a(X_2) (H_f+1)^{-1}  \right\|^2 }
\\
&& \leq \int \frac{dX^{(1,1)}}{|X^{(1,1)}|^2} \sup_{r \geq 0} \left\| e^{\gamma_{l+1} \langle x \rangle } {w}^{(I,j_l)}_{1,1}(X_1,\widetilde{X}_2 )
e^{-\gamma_{l} \langle x \rangle } \right\|_{\HH_{\rm at} \to \HH_{\rm at}} \frac{(r + |X_1|)(r + | \widetilde{X}_2|)}{(r + |\widetilde{X}_2|)^2}
\\
 &&\leq c |g|^{2} ,
\end{eqnarray*}
and similarly if $p_l=2,q_l=0$ or $p_l=0$, $q_l=2$. Since
\begin{eqnarray*}
\| (-\Delta + 1 )^{1/2} \otimes (H_f + 1 )^{1/2} ( - \Delta \otimes 1 + 1 \otimes H_f + 1 )^{-1} \| &=& 1 \\
\| 1 \otimes (H_f + 1 )  ( - \Delta \otimes 1 + 1 \otimes H_f + 1 )^{-1} \| &=& 1
\end{eqnarray*}
we have proved \eqref{initial:thmE22:eq1A}.
 A similar argument gives \eqref{initial:thmE22:eq2}
\begin{equation}
\left| \partial_r \partial_\beta^s V_{\umm,\upp,\unn,\uqq}[w^{(I)}](r,K^{|\umm|,|\unn|})\right| \leq L^{k+1}C^L  | g|^{|\umm| + |\unn| + |\upp| + |\uqq|} .
\end{equation}
One can use  the same estimates as above
to  show that
the $\beta$ derivative in  \eqref{eq:derVbetaj5-2} exists in
 $L^\infty(\underline{B}_1^{(|\umm|,|\unn|)} ; C^1([0,1]; \mathcal{B}(\HH)))$.
To show this one
 replaces ${w}^{(I,j_l)}$ by
its difference from the differential quotient, i.e.,
$(\Delta \beta)^{-1}(  {w}^{(I,j_l-1)}(\beta +  \Delta \beta) - {w}^{(I,j_l-1)}(\beta))  - {w}^{(I,j_l)}(\beta) $ and using the explicit expressions for $w^{(I)}$ it is straight
forward to verify using  Taylor's theorem with remainder that the right hand side in
the corresponding estimates
 converge to  zero  as $\Delta \beta$ tends to zero. Likewise
one shows continuity in $\beta$.
It now follows that the  $\beta$ derivative in  \eqref{eq:derVbetaj5-1} exists in
$L^\infty(\underline{B}_1^{(|\umm|,|\unn|)} ; C^1[0,1])$
The mapping $(g,z) \mapsto V_{\umm,\upp,\unn,\uqq}[w^{(I)}(g,E_{\rm at} + z)]$ is
in $C^{\omega}_B(\mathbb{C} \times D_{1/2}
; C_B^k(\R ; \WW_{|\umm|,|\unn|}^\#))$. To this end, observe that for fixed  $z$,
 $ V_{\umm,\upp,\unn,\uqq}$ is a polynomial in
$g$ with  coefficients in $C_B^{\omega,k}(\R ; \mathcal{W}_{|\umm|,|\unn|}^{\#})$.
For fixed  $g$  it is straight forward to verify that  $ V_{\umm,\upp,\unn,\uqq}$
 is differentiable
 with respect to $z$. To this end observe that  only $w^{(I)}_{0,0}$ depends on $z$.  

\end{proof}

Using Lemma   \ref{initial:thmE22} the proof of Theorem  \ref{thm:inimain1} (a) is analogous to
the proof of Theorem 17 (a) in \cite{HH10-2}. Below we summarize the main estimates of  the proof.
Let $S^L_{M,N}$ denote the set of tuples $(\umm,\upp,\unn,\uqq) \in \N_0^{4L}$ with
$|\umm|=M$, $|\unn|=N$, and $1 \leq m_l+p_l+q_l+n_l \leq 2$.
We find, with $\widetilde{\xi} := (8 \pi)^{-1/2} \xi$,
\begin{align}
\| w^{(0)}_{\geq 1}(g,z) \|_\xi^{(\#,k)} &= \sup_{\beta \in \R} \sum_{M+N \geq 1} {\xi}^{-(M+N)}
\max_{0 \leq l \leq k}
 \| \partial_\beta^l \widetilde{w}_{M,N}(g,\beta, z) \|^{\#} \nonumber \\
&\leq \sum_{M+N\geq 1} \sum_{L=1}^\infty \sum_{(\umm,\upp,\unn,\uqq) \in S^L_{M,N}}
\widetilde{\xi}^{-(M+N)} 4^L
 \sup_{\beta \in \R} \max_{0 \leq l \leq k}
 \| \partial_\beta^l
 V_{\umm,\upp,\unn,\uqq}[w^{(I)}(g,\beta, \zeta)] \|^{\#} \nonumber \\
&\leq \sum_{L=1}^\infty \sum_{M+N\geq 1} \sum_{(\umm,\upp,\unn,\uqq) \in S^L_{M,N} }
\widetilde{\xi}^{-|\umm|-|\unn|} L^{k+1} (4 C)^L g^{|\umm|+|\unn|+|\upp|+|\uqq|}  \nonumber \\
&\leq\sum_{L=1}^\infty  L^{k+1}  14^L \widetilde{\xi}^{-2L} \left( 4 C |g|  \right)^L \; , \label{eq:feb2:1}
\end{align}
for all $(g,z) \in D_1 \times D_{1/2} $.
A similar but simpler estimate yields
\begin{align}
\sup_{r \in [0,1]} \| \partial_r w^{(0)}_{0,0}(g,z)(r) - 1 \|_{C^k(\R)}
&\leq \sum_{L=2}^\infty \sum_{(\upp,\uqq) \in \N_0^{2L}: p_l+q_l= 1, 2}
\sup_{\beta \in \R } \max_{0 \leq l \leq k} \| \partial_\beta^l  V_{\uzz,\upp,\uzz,\uqq}[w^{(I)}(g,\beta,\zeta)] \|^{\#} \nonumber \\
&\leq \sum_{L=2}^\infty 3^L L^{k+1} \left( C |g|  \right)^L \; , \label{eq:feb2:2}
\end{align}
for all $(g,z) \in D_1 \times D_{1/2}$.
Analogously we have for all  $(g,z) \in D_1 \times D_{1/2}$,
\begin{align}
\| w^{(0)}_{0,0}(g,z)(0) + z \|_{C^k(\R)}
& \leq \sum_{L=2}^\infty \sum_{(\upp,\uqq) \in \N_0^{2L}: p_l+q_l= 1, 2}
\sup_{\beta \in \R } \max_{0 \leq l \leq k} \| \partial_\beta^l  V_{\uzz,\upp,\uzz,\uqq}[w^{(I)}(g,\zeta)] \|^{\#} \nonumber
\\ \label{eq:feb2:3}
& \leq \sum_{L=2}^\infty 3^L  L^{k+1}  \left(  C |g|  \right)^L .
\end{align}
The right hand sides in \eqref{eq:feb2:1}--\eqref{eq:feb2:3} can be made arbitrarily small for sufficiently small $|g|$.
This implies that $w^{(0)}(g)$ is in  $\mathcal{B}^{(k)}(\delta_1,\delta_2,\delta_3)$. Rotation invariance and the symmetry property have already been shown  in    Theorem 17 of \cite{HH10-2}.
Theorem \ref{thm:inimain1}  (c)   follows from  Lemma   \ref{initial:thmE22}
and   the  convergence for small $g$ established in  \eqref{eq:feb2:1}--\eqref{eq:feb2:3}.

\section{Renormalization Transformation}
\label{sec:ren:def}

In this section we define the Renormalization transformation as in \cite{BCFS03}. It is a combination of
the Feshbach transformation which cuts out higher photon energies, a rescaling of the resulting operator so that it acts on the fixed subspace $\HH_{\rm red}$ and a conformal transformation of the spectral parameter.
Let  $0<\xi<1$ and $0 < \rho < 1$.
For $w \in \mathcal{W}_\xi$ we  define the analytic function
$$E_\rho[w](z) := \rho^{-1} E[w](z) :=   - \rho^{-1} \langle \Omega , H(w(z)) \Omega \rangle$$
 and the set
$$
U[w] := \{ z \in D_{1/2}  | | E[w](z) | <  \rho  / 2 \} .
$$
\begin{lemma} \label{renorm:thm3} Let $0< \rho\leq 1/2$.
Then for all $w \in \mathcal{B}(\rho/8, \rho/8,  \rho/8 )$, the
function  $E_{\rho}[w]: U[w] \to D_{1/2}$ is an analytic  bijection,
 $D_{3\rho/8} \subset U[w] \subset D_{5\rho/8}$, and for all $z \in D_{5 \rho/8}$ we have
\begin{equation} \label{estonpartialE}
|\partial_z E[w](z) - 1 | \leq \frac{4 \rho }{( 4  - 5\rho)^{2}} .
\end{equation}
If   $w \in \mathcal{B}(\rho/32, \rho/32,  \rho/32 )$, then
$D_{15   \rho/32} \subset U[w] \subset D_{17\rho/32}$ and for all $z \in D_{17 \rho/32 }$ we have
\begin{equation} \label{estonpartialE2}
|\partial_z E[w](z) - 1 | \leq \frac{16 \rho }{( 16   - 17 \rho)^{2}} .
\end{equation}

\end{lemma}
For a proof of the lemma we apply  following lemma with
$r = \rho/2$ and $\epsilon = \rho/8$ respectively $\epsilon = \rho/32$.  For a proof of Lemma  \ref{renorm:thm3A}
see  \cite{HH10-1} (Lemma 22) or \cite{BCFS03}.
\begin{lemma} \label{renorm:thm3A} Let $0 < \epsilon < 1/2$, and
let $E : D_{1/2} \to \C$ be an analytic function which satisfies
\begin{align*}
\sup_{z \in D_{1/2}} |E(z) - z | \leq \epsilon .
\end{align*}
Then for any $r > 0 $ with $r + \epsilon < 1/2$ the following is true.
\begin{itemize}
\item[(a)] For $w \in D_{r}$ there exists a unique $z \in D_{1/2}$ such that
$E(z) = w$.
\item[(b)] The map $E :  U_r := \{ z \in D_{1/2} | | E(z) | < r \} \to D_{r} $ is
biholomorphic.
\item[(c)]  We have $D_{r - \epsilon} \subset U_r \subset D_{r+\epsilon}$.
\item[(d)] If  $z \in  D_{r + \epsilon }$, then $|\partial_z E(z) - 1 | \leq \frac{\epsilon}{2} ( 1/2 - (r + \epsilon))^{-2}$.
\end{itemize}
\end{lemma}
If $0 < \rho \leq 1/4$, then for  $w \in \mathcal{B}(\rho/32, \rho/32,  \rho/32 )$
we find using   \eqref{estonpartialE2}, that for all $z \in D_{17  \rho /32}$
\begin{equation} \label{eq:pzeestimate}
| \partial_z E_\rho[w] | \geq    \frac{1}{\rho} ( 1  -  | \partial_z E - 1   | ) \geq \frac{15}{16 \rho} .
\end{equation}
 Let $I_\rho[w]$ denote the inverse of $
E_{\rho}[w]: U[w] \to D_{1/2}$. It satisfies
\begin{equation} \label{eq:identityrelEI}
E_\rho[w] (I_\rho[w](z)) = z ,
\end{equation}
for all $z \in D_{1/2}$. For notational compactness we shall occasionally drop the dependence on $w$ and
write $E_\rho$ and $I_\rho$.
In the previous section we introduced  smooth functions  $\chi_1$ and  $\chib_1$. We set
$$
\chi_\rho(\cdot)  = \chi_1(\cdot /\rho) \quad , \quad \chib_\rho(\cdot) = \chib_1(\cdot /\rho) \; ,
$$
and use the abbreviation $\chi_\rho = \chi_\rho(H_f)$ and $\chib_\rho  = \chib_\rho(H_f)$.
It should be clear from the context whether $\chi_\rho$ or $\chib_\rho$ denotes a function or an operator.
The following theorem is proven in \cite{BCFS03,HH10-1}.

\begin{lemma} \label{renorm:thm1} Let $0 < \rho \leq  1/2$.  Then for all $w \in \mathcal{B}(\rho/8,\rho/8,\rho/8)$, and
all $z \in D_{1/2}$ the pair of operators  $(H(w(E_\rho^{-1}(z)),H_{0,0}(E_\rho^{-1}(z)))$ is a Feshbach pair for $\chi_\rho$.
\end{lemma}

The definition of  the renormalization transformation involves a  scaling transformation $S_\rho$ which scales the energy value $\rho$ to the value 1.
For operators  $A \in \mathcal{B}(\FF)$ we define
$$
S_\rho(A) =  \rho^{-1}   \Gamma_\rho A \Gamma_\rho^* ,
$$
where  $\Gamma_\rho$ is the unitary dilation on $\FF$ which is  uniquely determined  by
\begin{align*}
& \Gamma_\rho a^\#(k) \Gamma_\rho^* = \rho^{-3/2} a^\#(\rho^{-1} k) , \qquad \Gamma_\rho \Omega = \Omega .
\end{align*}
It is easy to check that $\Gamma_\rho H_f \Gamma_\rho^* = \rho H_f$ and hence $\Gamma_\rho \chi_\rho \Gamma_\rho^* = \chi_1$.
We are now ready to  define the renormalization transformation, which  in
view of Lemmas \ref{renorm:thm3} and \ref{renorm:thm1} is well defined.

\begin{definition} Let $0 < \rho \leq 1/2$. For $w\in \mathcal{B}(\rho/8,\rho/8,\rho/8)$ we define the renormalization transformation
$$
\left( \mathcal{R}_\rho H(w) \right)(z) :=  S_\rho F_{\chi_{\rho}}(H(w(E_\rho^{-1}(z)),H_{0,0}(E_\rho^{-1}(z))) \upharpoonright \HH_{\rm red} ,
$$
where  $z \in D_{1/2}$.
\end{definition}

\begin{theorem} \label{thm:maingenerala} Let $0<\rho \leq 1/2$ and $0< \xi \leq 1/2$.
For
$w\in \mathcal{B}(\rho/8,\rho/8,\rho/8)$ there exists a unique integral kernel $\mathcal{R}_\rho(w) \in \WW_{\xi}$
$$
(\mathcal{R}_\rho H(w))(z) = H(\mathcal{R}_\rho(w)(z)) .
$$
If $w$ is symmetric then also $\mathcal{R}_\rho(w)$ is symmetric. If $w(z)$ is invariant under rotations for all $z \in D_{1/2}$
than also $\mathcal{R}_\rho(w)(z)$ is invariant under rotations for all $z \in D_{1/2}$.
\end{theorem}

A proof of  the existence of the integral kernel as stated in  Theorem \ref{thm:maingenerala} can be found in \cite{BCFS03}
or  \cite{HH10-1} (Theorem 32). The uniqueness follows from Theorem  \ref{thm:injective}. The statement about the symmetry and the  rotation invariance
follows from  Lemmas  \ref{lem:symmetry} and \ref{lem:wHinvequiv}  and the fact  that the renormalization transformation preserves
symmetry and rotation invariance, respectively.  This is explained in detail in \cite{HH10-2}.
The renormalized kernels are given as follows.
 For  $w \in \WW_{m+p,n+q}^\#$ we define
\begin{eqnarray*}
\lefteqn{ W_{p,q}^{m,n}[w]( r , K^{(m,n)} ) } \\
&:=& P_{\rm red} \int_{B_1^{p+q}} \frac{d X^{(p,q)}}{|X^{(p,q)}|^{1/2}} a^*(x^{(p)}) w_{p+m,q+n}(H_f + r , x^{(p)},
k^{(m)}, \widetilde{x}^{(q)}, \widetilde{k}^{(n)} )
 a(\widetilde{x}^{(q)}) P_{\rm red}
\end{eqnarray*}
which defines an operator for a.e. $K^{(m,n)} \in B_1^{m+n}$. In the case $m = n = 0$ we set   $W^{0,0}_{m,n}[w](r) := W_{m,n}[w](r)$.
  For $w \in \mathcal{B}(\rho/8, \rho/8,\rho/8)$  we have
\begin{equation} \nonumber
\mathcal{R}_\rho(w)(z) = \mathcal{R}_\rho^\#(w(I_\rho[w](z)))  \; .
\end{equation}
For $w \in \WW_\xi^\#$ we define
\begin{equation} \nonumber 
\mathcal{R}_\rho^\#(w) := \widehat{w}^{({\rm sym})} \; ,
\end{equation}
where the  kernels  $\widehat{w}$ are given as follows.
For $M + N \geq 1$,
\begin{align} \label{eq:renormalizedkernels01}
\widehat{w}_{M,N}( r , K^{(M,N)} ) & :=
\sum_{L=1}^\infty (-1)^{L-1} \rho^{M+N - 1}  \sum_{ \substack{ (\umm,\upp,\unn,\uqq) \in \N_0^{4L}:  \\    |\umm|=M, |\unn|=N , \\  m_l+p_l+n_l+q_l \geq 1 }}    \\
& \ \     \prod_{l=1}^L \left\{ \binom{ m_l + p_l }{ p_l} \binom{ n_l + q_l }{ q_l}    \right\} v_{\umm,\upp,\unn,\uqq}[w](r, K^{(M,N)}) , \nonumber
\end{align}
and
\begin{align} \label{eq:renormalizedkernels00}
\widehat{w}_{0,0}( r ) &:= \rho^{-1} w_{0,0}(\rho r)  + \rho^{-1} \sum_{L=2}^\infty (-1)^{L-1}
\sum_{ \substack{ (\upp,\uqq) \in \N_0^{2L}: \\ p_l + q_l \geq 1 } }
v_{\underline{0},\upp, \underline{0},\uqq}[w](r) \; .
\end{align}
Moreover, we have introduced  the expressions
\begin{eqnarray}
\lefteqn{ v_{\umm,\upp,\unn,\uqq}[w]( r, K^{(|\umm|,|\unn|)}) :=  } \label{eq:defofv} \\
&&
\left\langle \Omega  ,  F_0[w] (H_f + \rho (r + \widetilde{r}_0) )
\prod_{l=1}^L \left\{  {W}_{p_l,q_l}^{m_l,n_l}[w](\rho(r+r_l),
\rho K_l^{(m_l,n_l)} ) F_l[w]( H_f + \rho( r + \widetilde{r}_l ) ) \right\} \Omega
\right\rangle   , \nonumber
\end{eqnarray}
where $F_0[w](r) := \chi_\rho(r  )$ and  $F_L[w](r) := \chi_\rho(r  )$, and for $l = 1, ... , L - 1$
\begin{eqnarray} \label{eq:defofFFFF}
F_l[w](r)  := F[w](r) := \frac{ \overline{ \chi}_\rho^2(r )}{ w_{0,0}(r) } \; .
\end{eqnarray}
We used the notation introduced in
\eqref{eq:rldef} and
\eqref{eq:rltildedef}.
The next theorem states the contraction property.

\begin{theorem} \label{codim:thm1} For any positive numbers $\rho_0 \leq 1/4$ and $\xi_0 \leq 1/2$ there exist numbers $\rho, \xi, \epsilon_0$ satisfying
$\rho \in (0, \rho_0]$, $\xi \in (0, \xi_0]$, and $0 < \epsilon_0 \leq \rho/8$
such that the following property holds,
\begin{equation} \label{codim:thm1:eq}
\mathcal{R}_\rho : \mathcal{B}_0(\epsilon, \delta_1, \delta_2 ) \to
\mathcal{B}_0( \epsilon + \delta_2/2 , \delta_2/2 , \delta_2/2 ) \quad , \quad \forall \ \epsilon, \delta_1, \delta_2 \in [0, \epsilon_0) .
\end{equation}
\end{theorem}

A proof  of Theorem  \ref{codim:thm1} can be found in \cite{HH10-1} (Theorem 38).
The proof given there  relies on the  fact that there are no terms which are linear in
creation or annihilation operators. Since by rotation invariance and Lemma \ref{lem:wHinvequiv}
there are no terms which
are linear in creation and annihilation operators,  Theorem  \ref{codim:thm1} follows from
the same proof.
The contraction property allows us to iterate the renormalization transformation.  To this end we
introduce the following Hypothesis.

\vspace{0.5cm}

\noindent
($R$) \quad Let $\rho, \xi, \epsilon_0$ are positive numbers such that the contraction property \eqref{codim:thm1:eq}
holds and $\rho \leq 1/4$, $\xi \leq 1/4$ and $\epsilon_0 \leq \rho/8$.

\vspace{0.5cm}

Now we extend the renormalization transformation to  $\mathcal{B}^{(0)}(\rho/8,\rho/8,\rho/8)$ by setting
$$
\mathcal{R}_\rho(w )(\beta) = \mathcal{R}_\rho(w (\beta) )
$$
for $w \in \mathcal{B}^{(0)}(\rho/8,\rho/8,\rho/8)$ and
$$
\mathcal{R}_\rho^\#(w )(\beta) = \mathcal{R}_\rho^\#(w (\beta) )
$$
for $w \in \mathcal{B}^{(\#,0)}(\rho/8,\rho/8,\rho/8)$. That is we have
$$
\mathcal{R}_\rho(w )(\beta,z) = \mathcal{R}_\rho^\#(w (\beta,I_\rho(\beta,z) ) )
$$

The next theorem states that the extended renormalization transformation preserves the  $\mathcal{B}_0^{(k)}$-balls
and acts as a contraction on these balls in all but one dimension.

\begin{theorem} \label{contk1} For $k \in \N_0$ and positive numbers $\rho_0 \leq 1/4$ and $\xi_0 \leq 1/4$
there exists numbers
$\rho,\xi,\epsilon_0$ satisfying $\rho \in (0,\rho_0]$, $\xi \in (0,\xi_0]$, and $0 < \epsilon_0 \leq \rho/32$ such
that
\begin{equation} \label{codim:thm1:eq:ck}
\mathcal{R}_\rho : \mathcal{B}_0^{(k)}(\epsilon, \delta_1,\delta_2) \to \mathcal{B}_0^{(k)}(\epsilon +
\delta_2/4 + \delta_1/4,\delta_2/2,\delta_2/2) \quad , \quad \forall \epsilon , \delta_1 ,\delta_2 \in [0, \epsilon_0 ) .
\end{equation}
\end{theorem}

Theorem  \ref{contk1} will be shown below.
The next theorem states that the extended renormalization transformation preserves analyticity.

\begin{theorem} \label{thm:analytbasic}  Let $0<\rho \leq 1/2$ and $0< \xi \leq 1/2$. Let $S$ be an open  subset of $\C^\nu$ with $\nu \in \N$.
Suppose the map $w(\cdot,\cdot) : S \times \R  \to \WW_\xi^{\#}$ is in $C^{\omega,k}(S \times \R; \WW_\xi^{\#} ) $
and for all $s \in S$ we have $w(s,\cdot) \in \mathcal{B}^{(\#,k)}(\rho/32, 5 \rho/ 8 ,\rho/32)$. Then
$$
(s,\beta) \mapsto \mathcal{R}^\#_\rho(w(s,\beta))
$$
is in  $C^{\omega,k}_B(S \times \R; \WW_\xi^{\#} ) $.
\end{theorem}


\begin{theorem} \label{thm:analytbasic2}  Let $0<\rho \leq 1/2$ and $0< \xi \leq 1/2$.
Let $S$ be an open  subset of $\C$. Suppose
\begin{eqnarray*}
w(\cdot, \cdot,\cdot ) : && S \times D_{1/2} \times \R \to \WW_\xi^{\#} \\
&&(s,z,\beta) \mapsto w(s,z,\beta)
\end{eqnarray*}
is in $C^{\omega,k}(  S \times D_{1/2}  \times \R ;   \WW_\xi^{\#})$  and for all $s \in S$ we have
 $w(s,\cdot, \cdot ) \in \mathcal{B}^{(k)}(\rho/32,\rho/32,\rho/32)$. Then
$$
(s,z,\beta ) \mapsto (\mathcal{R}_\rho(w(s,\cdot,\beta)))(z)
$$
is  in $C^{\omega,k}_B(  S \times D_{1/2}  \times \R ;   \WW_\xi^{\#})$.
\end{theorem}

To show Theorems \ref{contk1},  \ref{thm:analytbasic}, and   \ref{thm:analytbasic2}  we will use the explicit expression for the
 renormalized integral kernels introduced above.
For  $w \in \mathcal{B}^{(0)}(\rho/8,\rho/8,\rho/8)$ we define
\begin{align*}
E_\rho(\beta,z) := E_\rho[ w{(\beta)}](z) , \quad
I_\rho(\beta,z)  := I_\rho[  w{(\beta)}](z) .
\end{align*}

The crucial point of that following  estimate is that the constant $C_L$  grows at most polynomially in $L$
and  that  $\rho^{-1}$ occurs to
a power of at most $L-1$.

\begin{lemma} \label{codim:thm3first} Let $0 \leq \rho \leq 1/4$ and
  let   $w \in \mathcal{B}^{(\#,k)}(\rho/32,5 \rho/8,\cdot)$.
Then for $(\umm,\upp,\unn,\uqq) \in (\N_0^{L})^4$ we  have
\begin{equation} \label{eq:vestimate2}
 \max_{0 \leq l \leq k} || \partial_\beta^l  v_{\umm,\upp,\unn,\uqq}[w(\beta)] ||^{\#}   \leq
 C_L  \left( \frac{1}{t}\right)^{L-1}  \prod_{l=1}^{L}
    \frac{ \max_{0 \leq l \leq k} \| \partial_\beta^l  w_{m_l + p_l, n_l + q_l}(\beta) \|^{\#} }  { \sqrt{p_l ! q_l !}}  ,
\end{equation}
where
$t := 3 \rho /32$ and  $C_L$ is a constant which satisfies a bound
$$
C_L \leq c( 1 +  \| \partial_r \chi_1 \|_\infty)^k   ( 1 + L^k ) ,
$$
where $c$ is a finite numerical constant.
\end{lemma}

\begin{proof}
First we consider the case $k=0$. Since in that case the $\beta$ dependence is not relevant we
drop the  $\beta$ dependence in the notation.
Using
\begin{equation} \label{eq:babyestimate1}
| \langle \Omega , A_1 A_2 \cdots  A_n \Omega \rangle | \leq \| A_1 \|_{\rm op} \| A_2 \|_{\rm op}  \cdots \| A_n \|_{\rm op},
\end{equation}
we find
\begin{eqnarray*}
\lefteqn{ \esssup_{ K^{(|\umm|,|\unn|)}} \sup_{r \in [0,1]}| v_{\umm,\upp,\unn,\uqq}[w](r, K^{(|\umm|,|\unn|)})  | } \\
&&
\leq   \prod_{l=1}^{L}  \esssup_{K^{(m_l,n_l)}} \sup_{r \in [0,1]}  \| W_{p,q}^{m,n}[w](r,K^{(m_l,n_l)})\|_{\rm op} ,
    \prod_{l=1}^{L-1} \| \chib_\rho^2/w_{0,0} \|_{C[0,1]} .
\end{eqnarray*}
To estimate the right hand side we use
\begin{eqnarray} \label{eq:ineqforW1}
&&\esssup_{K^{(m,n)}  }  \sup_{r \in [0,1]}  \|   W_{p,q}^{m,n}[w](r,K^{(m,n)}) \|_{\rm op} \leq
\frac{ \|   w_{p+m,q+n}, \|_{L^\infty(B_1^{m+n};C[0,1])}  }{\sqrt{p! q!}} \\
&& \| \chib_\rho^2/w_{0,0} \|_{C[0,1]} \leq 1/t \label{eq:ineqforH1} .
\end{eqnarray}
Inequality  \eqref{eq:ineqforW1} can be shown using Lemma   \ref{kernelopestimate}
and \eqref{eq:intofwKminus2}. Inequality  \eqref{eq:ineqforH1} can be shown as follows.
For $ r \geq \rho 3/4$ we have
\begin{eqnarray*}
| w_{0,0}(r)| \geq r - | r - ( w_{0,0}(r) - w_{0,0}(0))|  - | w_{0,0}(0)| \geq r - r \frac{\rho}{32} - 5 \rho /8 \geq
\rho \frac{3}{32} ,
\end{eqnarray*}
and thus
\begin{equation}  \label{eq:ineqforH11}
\left[ {\rm inf}_{r \in [\rho \frac{3}{4},1]} |w_{0,0}(r) | \right]^{-1} \leq 1/t .
\end{equation}
Next we calculate the derivative with respect to $r$. To this end first observe that
using Lemma   \ref{kernelopestimate} and dominated convergence one can show  that
for a.e. $K^{(m,n)}$ the partial derivative  $ \partial_r W^{m,n}_{p,q}[w](r,K^{(m,n)})$
exists with respect to the operator norm topology and equals  $  W^{m,n}_{p,q}[\partial_r w](r,K^{(m,n)})$.
Thus
\begin{eqnarray} \label{eq:ineqforW11}
\esssup_{K^{(m,n)}  }  \sup_{r \in [0,1]}  \| \partial_r  W_{p,q}^{m,n}[w](r,K^{(m,n)}) \|_{\rm op} \leq
\frac{ \| \partial_r   w_{p+m,q+n}, \|_{L^\infty(B_1^{m+n};C[0,1])}  }{\sqrt{p! q!}} .
\end{eqnarray}
Furthermore,
\begin{equation*}
D_r \frac{\chib_\rho^2}{w_{0,0}} =    - \frac{\chib_\rho^2}{w_{0,0}^2} (\partial_r w_{0,0})     +
 \frac{2 \chib_\rho \partial_r \chib_\rho}{w_{0,0} }
\end{equation*}
and thus for $s + \rho r \in [0,1]$ we have
\begin{equation}  \label{eq:ineqforH111}
| D_r \frac{\chib_\rho^2}{w_{0,0}}(s + \rho r)  | \leq \frac{3}{2} \frac{ \rho}{ t^2} + \frac{2 \| \chi_1' \|_{\infty}}{t} ,
\end{equation}
where we used  $\| \partial_r w_{0,0} \|_{C[0,1]} \leq 3/2$.
Calculating the derivative with respect to $r$ using Leibniz and estimating
the resulting expression with the help of \eqref{eq:babyestimate1},  \eqref{eq:ineqforW1}
\eqref{eq:ineqforH1}, \eqref{eq:ineqforW11}, and \eqref{eq:ineqforH111}  the Inequality  \eqref{eq:vestimate2}
follows for $k=0$.

Next we show \eqref{eq:vestimate2} for $k \geq 1$.
It follows from Lemma \ref{lem:diff} (b) that
$ \beta \mapsto \frac{\chib_\rho^2}{w_{0,0}(\beta)} $ is in $C^k(\R , \WW_{0,0}^\# )$.
We  use \eqref{eq:faadibruno} to calculate the derivative of
${\chib}_\rho^2/w_{0,0}(\beta)$ with respect to $\beta$,
\begin{eqnarray} \label{eq:parbetaF}
\partial_\beta^l \frac{{\chib}_\rho^2}{w_{0,0}(\beta)} =
\sum_{X \in P_l} |X|! (-1)^{|X|} \frac{ \chib_\rho^2}{(w_{0,0}(\beta))^{|X|+1}} \prod_{x \in X} \partial_\beta^{|x|} w_{0,0}(\beta)  .
\end{eqnarray}
The derivative  in \eqref{eq:parbetaF} is with respect to the $C[0,1]$ norm.
To estimate  the right hand side of \eqref{eq:parbetaF}  we use  \eqref{eq:ineqforH11}
 that by assumption  $\| \partial_\beta^j w_{0,0}(\beta) \|_{C[0,1]} \leq 5 \rho/8$.
It follows that there exits a finite constant, $C_{F,l}$, independent of $\rho$
such that
\begin{equation} \label{eq:estonbetaderofw00}
\left\| \partial_\beta^l  \frac{{\chib}_\rho^2}{w_{0,0}(\beta)} \right\|_{C[0,1]} \leq \frac{C_{F,l}}{t}  ,
\end{equation}
and   ${C_{F,0}}=1$.  Using \eqref{eq:faadibruno} we find
\begin{eqnarray}
 \lefteqn{
D_r \partial_\beta^l    \frac{{\chib}_\rho^2}{w_{0,0}(\beta)} } \nonumber \\
 &&=
\sum_{X \in P_l} |X|! (-1)^{|X|} \frac{ \chib_\rho}{(w_{0,0}(\beta))^{|X|+1}}  D(w_{0,0}(\beta),|X|,\chib_1,\rho) \prod_{x \in X} \partial_\beta^{|x|} w_{0,0}(\beta) \nonumber
 \\
&&
+
\sum_{X \in P_l} |X|! (-1)^{|X|} \frac{ \chib_\rho}{(w_{0,0}(\beta))^{|X|+1}}  \sum_{x \in X}
( \partial_r \partial_\beta^{|x|}  w_{0,0}(\beta) )  \prod_{x' \in X, x' \neq x} \partial_\beta^{|x'|} w_{0,0}(\beta) , \label{eq:diffofprpbeta}
\end{eqnarray}
where we wrote
$$
D(w_{0,0},m,\chib_1,\rho) := \frac{ 2  \partial_r \chib_1(\cdot /\rho) }{\rho} -
( m +1) \frac{\chib_\rho}{w_{0,0}} \partial_r w_{0,0} .
$$
We estimate
\begin{equation} \label{eq:Destimate}
\| D(w_{0,0},m,\chib_1,\rho) \|_\infty \leq \frac{2}{\rho}\| \partial_r \chib_1 \|_\infty + (m+1) \frac{8}{\rho},
\end{equation}
where we used that by assumption it follows that $\| \partial_r w_{0,0} \|_\infty \leq 3/2$.
The derivative  in  \eqref{eq:diffofprpbeta} is with respect to the $C[0,1]$ norm.
Inserting
\eqref{eq:Destimate}  into \eqref{eq:diffofprpbeta} we find for $s + \rho r \in [0,1]$
\begin{equation} \label{eq:estonbetaderofw00derr}
\left|  D_r \partial_\beta^l \left(  \frac{\chib_{\rho}^2}{w_{0,0}(\beta)}(s + \rho r )\right)  \right|
\leq
t^{-1} C_{F,l}  ( 2 \| \partial_r \chib_1\|_\infty + ( l + 1 ) 8 ) +  t^{-1} l C_{F,l} .
\end{equation}

Next  observe that $v_{\umm,\upp,\unn,\uqq}[\cdot]$ is given as
a multilinear expression
of kernels $(w_{m,n})_{m+n\geq 1}$ and $\frac{\chib_\rho}{w_{0,0}}$.
It follows from Lemma \ref{lem:diff} that
$\beta \mapsto   v_{\umm,\upp,\unn,\uqq}[w(\beta)]$ is in  $C^k(\R; \WW_{|\umm|,|\unn|})$
and that  Leibniz rule for higher derivatives  \eqref{eq:higherleibniz}  is applicable
to calculate derivatives  $D_\beta^l v_{\umm,\upp,\unn,\uqq}[w(\beta)]$.
We thus apply   \eqref{eq:higherleibniz} and estimate the resulting expression using  \eqref{eq:babyestimate1}.
To this end we use
\begin{eqnarray} \label{eq:ineqforW1111}
\esssup_{K^{(m,n)}  }  \sup_{r \in [0,1]} \sum_{s=0}^1
 \| \partial_r^s  W_{p,q}^{m,n}[ \partial_\beta^l w](r,K^{(m,n)}) \|_{\rm op} \leq
\frac{ \| \partial_\beta^l   w_{p+m,q+n}  \|^\# }{\sqrt{p! q!}} ,
\end{eqnarray}
which follows from
 \eqref{eq:ineqforW1} and  \eqref{eq:ineqforW11}.
Using
 \eqref{eq:estonbetaderofw00},
\eqref{eq:estonbetaderofw00derr}, and \eqref{eq:ineqforW1111} Inequality \eqref{eq:vestimate2} now follows
from the following observation. The right hand side of \eqref{eq:higherleibniz}
contains $L^k$ terms. Each term contains at most $k$ factors involving a derivative.
\end{proof}

\vspace{0.5cm}

\noindent
{\it Proof of Theorem \ref{thm:analytbasic}}.
First observe that by Lemma \ref{lem:diff} (b)
\begin{equation} 
\left[ (s, \beta ) \mapsto \frac{\chib_\rho^2}{w_{0,0}(s, \beta)} \right] \in C^{\omega,k}(S \times \R , \WW_{0,0}^\# ) .
\end{equation}
It now follows from part (a) of the same Lemma that
the map  $(s, \beta)  \to v_{\umm,\upp,\unn,\uqq}[w(s,\beta)]$ is in
$C^{\omega,k}(S \times \R; \WW_{|\umm|,|\unn|}^{\#})$.
Using the estimate of Lemma   \ref{codim:thm3first} one can  show
the same way as  in  \cite{HH10-1} Theorem 31 that $\mathcal{R}_\rho^\#(w(s,\beta))$  is given
as a sum which is uniformly  convergent on subsets which constitute an open covering of $\R \times S$ and
that the sum is uniformly bounded.
This is done in Appendix F.
\qed

\begin{lemma}  \label{lem:estonbderofinverse}
Let $0 < \rho \leq 1/4$ and assume  $w \in \mathcal{B}^{(k)}(\cdot,\delta,\cdot)$, with $\delta \leq \rho/32$.
Then $I_\rho \in C_B^{k,\omega}(\R \times D_{1/2})$ and
\begin{align} \label{eq:lemonI1}
\sup_{(\beta,z) \in \R \times D_{1/2}} | \partial_z I_\rho(\beta,z) | &\leq     \frac{16 \rho}{15} .
\end{align}
Moreover, there exists a finite constant $C_k$ depending only on $k$, such that
\begin{align} \label{eq:lemonI2}
\max_{1 \leq s \leq k} \sup_{(\beta,z) \in \R \times D_{1/2}} | \partial_\beta^s I_\rho(\beta,z) | &\leq C_k \delta .
\end{align}
\end{lemma}
\begin{proof} The assumption $w \in \mathcal{B}^{(k)}(\cdot,\delta,\cdot)$ implies that
$E_\rho \in C^{k,\omega}(\R \times D_{1/2})$. By this  and  inequality  \eqref{eq:pzeestimate}
it follows from the inverse function theorem that $I_\rho$ is  in
$C^{k,\omega}(\R \times D_{1/2})$.  Let $(\beta,z) \in \R \times D_{1/2}$.
>From \eqref{eq:identityrelEI} we have
\begin{equation}  \label{eq:identityrelEI2}
E_\rho(\beta,I_\rho(\beta,z) ) = z .
\end{equation}
Differentiating  \eqref{eq:identityrelEI2} with
respect to $z$ we find
\begin{align*}
\partial_z I_\rho(\beta,z) = - \frac{1}{\partial_2 E_\rho(\beta, I_{\rho}(\beta, z) )}  ,
\end{align*}
where $\partial_i$ denotes the derivative with respect to the $i$-th argument (note that $\partial_1$ is a real
derivative and $\partial_2$ is a complex derivative).
By this and   \eqref{eq:pzeestimate} we obtain the bound \eqref{eq:lemonI1}.
Now we show the remaining bounds.  Differentiating  $\rho$  \eqref{eq:identityrelEI2} with
respect to $\beta$, we find
\begin{align} \label{eq:partialbeta1}
\partial_\beta I_\rho(\beta,z) = - \frac{ \partial_1 E(\beta,  I_{\rho}(\beta, z) )   }{\partial_2 E(\beta, I_{\rho}(\beta, z) )} ,
\end{align}
with $E(\beta, z) = \rho E_\rho( \beta, z)$.
This and \eqref{eq:pzeestimate} shows  \eqref{eq:lemonI2} for $k=1$. To show   \eqref{eq:lemonI2}  for $k \geq 2$ we proceed by
induction and use that   the assumption $w \in \mathcal{B}^{(k)}(\cdot,\delta,\cdot)$ implies
\begin{equation} \label{eq:boundonpartal1E}
| \partial_1^{s} E(\beta,z) | \leq \delta
\end{equation}
 for all $1 \leq s \leq k$.
Suppose    \eqref{eq:lemonI2}          holds for $k=n$. We then show that it holds for $k=n+1$.
We  differentiate \eqref{eq:partialbeta1} with respect to $\beta$. Using Leibniz  we obtain
\begin{align*} 
\partial_\beta^{n+1} I_\rho(\beta,z) =  \sum_{p=0}^n \binom{n}{p} A_p B_{n-p}  ,
\end{align*}
where
\begin{align*}
A_p &:= D_\beta^{p}  \partial_1 E(\beta, I_{\rho}(\beta, z) )  , \\
B_p &:= D_\beta^{p}  i(  \partial_2 E(\beta,  I_{\rho}(\beta, z) )  ,
\end{align*}
with $i(z) := - z^{-1}$. Now  using  \eqref{eq:faadibruno}, we find
\begin{align*}
A_p =  \sum_{q=0}^p \binom{p}{q} \sum_{X \in P_{q}}  \partial_1^{1+ p - q}   \partial_2^{|X|}  E(\beta, I_{\rho}(\beta, z) ) \prod_{x \in X} \partial_\beta^{|x|} I_\rho(\beta,z)
\end{align*}
Using \eqref{eq:boundonpartal1E}, analyticity of $E_\rho$ in the second argument,
and  the induction Hypothesis it follows that $|A_p| \leq C  \delta$ for some finite  constant,
$C$, depending only on $p$. To this end we note that derivatives $\partial_2$ can be estimated
using Cauchys formula and $\ran I_\rho \subset D_{17\rho/32}$, which follows
from    Lemma  \ref{renorm:thm3}.
Using  \eqref{eq:faadibruno} we find that
\begin{align*}
B_p =  \sum_{X \in P_p} (-1)^{|X|+1} |X|! \left( \partial_2 E(\beta, I_\rho(\beta, z) ) \right)^{-|X|-1} \prod_{x \in X} D_\beta^{|x|}
\partial_2 E(\beta, I_\rho(\beta,z)) .
\end{align*}
By \eqref{eq:pzeestimate} and  \eqref{eq:boundonpartal1E} we now see, similarly as for $A_p$, that
$|B_p | \leq C$ for some finite constant $C$  depending only on $p$.
\end{proof}

\vspace{0.5cm}

\noindent
{\it Proof of Theorem \ref{thm:analytbasic2}}.
By assumption it follows that $E_\rho \in C^{\omega,k}(S \times D_{1/2} \times \R)$. By the inverse
function theorem and  \eqref{eq:pzeestimate} it follows that
$I_\rho \in C^{\omega,k}(S \times D_{1/2} \times \R)$. Moreover by Lemma  \ref{renorm:thm3}
\begin{equation}
 \ran I_\rho \subset D_{17\rho/32} . \label{eq:subsetcond}
\end{equation}
For $\zeta \in D_{17 \rho /32}$ we have
\begin{equation}
\| w(s ,\zeta, \beta) \|_{C^k(\R)} \leq \| w(s, \zeta, \beta) + \zeta \|_{C^k(\R)} +
\| \zeta \|_{C^k(\R)} \leq  \frac{5 \rho}{8} .
\end{equation}
Thus we can apply Theorem  \ref{thm:analytbasic} for $w |_{S \times D_{17\rho/32} \times \R}$ and
conclude that
$$
(s , \zeta, \beta) \mapsto \mathcal{R}_\rho^\#(w(s ,\zeta, \beta) )
$$
is in $C^{\omega,k}(S \times D_{17\rho/32} \times \R ; \WW_\xi^\#)$. By  \eqref{eq:subsetcond}
it follows from the chain rule that
$$
(s,z,\beta) \mapsto     \mathcal{R}_\rho (w(s,\beta))(z) =
 \mathcal{R}_\rho^\#(w(s,\zeta,\beta) ) |_{
\zeta = I_\rho(s,z,\beta)}
$$
is in  $C^{\omega,k}(S \times D_{1/2} \times \R; \WW_\xi^\# )$.
\qed

\vspace{0.5cm}

Theorem \ref{contk1}, which is proven  in Section \ref{app:contraction},
 allows us to iterate the extended renormalization transformation on the extended balls.
Let us introduce the following Hypothesis.

\vspace{0.5cm}

\noindent
($R^{(k)}$) \quad  Let $\rho, \xi, \epsilon_0$ are positive numbers such that the contraction property \eqref{codim:thm1:eq:ck}
holds and $\rho \leq  1/4$, $\xi \leq 1/4$ and $\epsilon_0 \leq \rho/32$.

\vspace{0.5cm}

Recall that by Theorem \ref{contk1} and Theorem  \ref{codim:thm1}
there exists a  nonempty set of parameters for which the Hypothesis $(R)$ and $(R^{(k)})$ are  satisfied.

\begin{theorem} \label{thm:bcfsmain} Let $k \in \N_0$. Assume Hypothesis ($R$) and  ($R^{(k)}$).
Then for $\epsilon_0>0$ and $\rho>0$ sufficiently small there  exist functions
\begin{align*}
&e_{(0)}[ \cdot ] : \mathcal{B}_0(\epsilon_0/2,\epsilon_0/2,\epsilon_0/2) \to D_{1/2} \\
&\psi_{(0)}[ \cdot ] : \mathcal{B}_0(\epsilon_0/2,\epsilon_0/2,\epsilon_0/2) \to \FF
\end{align*}
such that the following holds.
\begin{itemize}
\item[(a)]  For all $w \in \mathcal{B}_0(\epsilon_0/2,\epsilon_0/2,\epsilon_0/2)$,
$$
{\rm dim} \ker \{ H(w(e_{(0)}[w]) \} \geq 1 ,$$
and $\psi_{(\infty)}[w]$ is a nonzero element in the kernel of  $H(w(e_{(0,\infty)}[w])$.
\item[(b)] If $w$ is symmetric and  $-1/2 <  z < e_{(0)}[w]$, then $H(w(z))$ is bounded invertible.
\item[(c)] The function $\psi_{(0)}[ \cdot ]$ is uniformly bounded with bound
$$\sup_{w \in \mathcal{B}_0(\epsilon_0/2,\epsilon_0/2,\epsilon_0/2)} \| \psi_{(0)}[w] \| \leq 4 e^4 . $$
If $H(w(z)) = H_f - z$, then $\psi_{(0)}[w] = \Omega$.
\item[(d)]  Suppose $w \in \mathcal{B}_0^{(k)}(\epsilon_0/2,\epsilon_0/2,\epsilon_0/2)$.
Then $\beta \to e_{(0)}[w (\beta)]$  and $\beta \to  \psi_{(0)}[w (\beta) ]$ are in $C_B^k(\R)$ and $C_B^k(\R; \FF)$, respectively.
\item[(e)] Let $S$ be an open  subset of $\C$. Suppose we are given a mapping
$(s,z,\beta) \mapsto  w(s,z,\beta)$
in $C_B^{\omega,k}(S \times D_{1/2} \times \R; \WW_\xi^\#)$ such that for all $s \in S$    we have $w(s,\cdot,\cdot) \in \mathcal{B}_0^{(k)}(\epsilon_0/2,\epsilon_0/2,\epsilon_0/2)$.
Then $s \mapsto   ( \beta \mapsto e_{(0)}[w(s)(\beta)])$  and $s \mapsto ( \beta \mapsto \psi_{(0)}[w(s)(\beta)])$ are $C^\omega_B(S ; C_B^k(\R))$and $C_B^\omega(S ; C_B^k(\R; \FF))$,
respectively.
\end{itemize}
\end{theorem}

Assumption ($R$) allows us to iterate the renormalization transformation as follows,
\begin{equation} \nonumber
\mathcal{B}_0(\sfrac{1}{2}\epsilon_0 , \sfrac{1}{2} \epsilon_0 , \sfrac{1}{2} \epsilon_0 ) \stackrel{\mathcal{R}_\rho}{\longrightarrow}
\mathcal{B}_0( [ \sfrac{1}{2}+ \sfrac{1}{4} ] \epsilon_0 , \sfrac{1}{4} \epsilon_0 , \sfrac{1}{4} \epsilon_0 ) \stackrel{\mathcal{R}_\rho}{\longrightarrow} \cdots
\mathcal{B}_0( \Sigma_{l=1}^n \sfrac{1}{2^l} \epsilon_0 , \sfrac{1}{2^n} \epsilon_0 , \sfrac{1}{2^n} \epsilon_0 )
 \stackrel{\mathcal{R}_\rho}{\longrightarrow} \cdots   .
\end{equation}
For $w \in \mathcal{B}_0(\epsilon_0/2,\epsilon_0/2,\epsilon_0/2)$ and $n \in \N_0$, we define
$$
w^{(n)} := \mathcal{R}_\rho^n ( w ) \in \mathcal{B}_0(  \epsilon_0 , \sfrac{1}{2^n} \epsilon_0 , \sfrac{1}{2^n} \epsilon_0 ) .
$$
We introduce the definitions
\begin{align*}
&E_{n,\rho}[w] := E_\rho[w^{(n)}] = \rho^{-1} E[w] \\
 & U_n[w] := U[w^{(n)}] := \{ z \in D_{1/2} | | E[w^{(n)}](z) | < \rho/2 \}
\end{align*}
By Lemma  \ref{renorm:thm3}  the map
$$
J_n[w] := E_{n,\rho}[w] : U_n[w] \to D_{1/2} , \quad z \mapsto E_{n,\rho}[w](z) .
$$
is an analytic bijection and $J_n[w]^{-1} : D_{1/2} \to U_n[w] \subset D_{1/2}$. For $0 \leq n \leq m$ we define
$$
e_{(n,m)}[w] := J_n[w]^{-1} \circ \cdots \circ J_m[w]^{-1}(0) .
$$
It has been shown in \cite{BCFS03}, see also \cite{HH10-1}, that the following limits exist
\begin{equation} \label{eq:defofe}
e_{(n,\infty)}[w] := \lim_{m \to \infty} e_{(n,m)}[w]
\end{equation}
We define the vectors in $\FF$,
of
$$
\psi_{(n,m)}[w] = Q_n[w] \Gamma_\rho^* Q_{n+1}[w] \Gamma_\rho^*  \cdots Q_{m-1} \Omega ,
$$
with
$$
Q_n[w] = \chi_\rho - \chib_\rho (H_n[w])_{\chib_\rho}^{-1} \chib_\rho W_n[w] \chi_\rho ,
$$
where
\begin{align*}
&H_n[w] := H(w^{(n)}(e_{(n,\infty)}[w]) ) \\
&T_n[w] := w_{0,0}^{(n)}(e_{(n,\infty)}[w])(H_f)  \\
&W_n[w] := H_n[w] - T_n[w] .
\end{align*}
It has been shown in \cite{BCFS03}, see also \cite{HH10-1}, that the following limit
exists
\begin{equation} \label{eq:defofpsi}
\psi_{(n,\infty)}[w] := \lim_{m \to \infty} \psi_{(n,m)}[w]
\end{equation}
and that $H_n[w] \psi_{(n,\infty)}[w]   = 0$. This implies part (a) of Theorem \ref{thm:bcfsmain}, with $e_{(0)}[w] = e_{(0,\infty)}[w]$ and
$\psi_{(0)}[w] = \psi_{(0,\infty)}[w]$. Part (b) has been shown in \cite{HH10-1}. Moreover, in \cite{HH10-1}, the bound
$\sup_{w \in \mathcal{B}_0(\epsilon_0/2,\epsilon_0/2,\epsilon_0/2)} \| \psi_{(0)}[w] \| \leq 4 e^4  $ was shown.
The second part of $(c)$ is a direct consequence of the definition of $\psi_{(0)}$.
Now let us show (d).
Assumption ($R^{(k)}$)
allows us to iterate the renormalization transformation as follows,
\begin{equation} \nonumber
\mathcal{B}_0^{(k)}(\sfrac{1}{2}\epsilon_0 , \sfrac{1}{2} \epsilon_0 , \sfrac{1}{2} \epsilon_0 )  \stackrel{\mathcal{R}_\rho}{\longrightarrow}
\mathcal{B}_0^{(k)}( [ \sfrac{1}{2}+ \sfrac{1}{4} ] \epsilon_0 , \sfrac{1}{4} \epsilon_0 , \sfrac{1}{4} \epsilon_0 )    \stackrel{\mathcal{R}_\rho}{\longrightarrow}    \cdots
\mathcal{B}_0^{(k)}( \Sigma_{l=1}^n \sfrac{1}{2^l} \epsilon_0 , \sfrac{1}{2^n} \epsilon_0 , \sfrac{1}{2^n} \epsilon_0 )   \stackrel{\mathcal{R}_\rho}{\longrightarrow} \cdots .
\end{equation}
We view
$w \in \mathcal{B}_0^{(k)}(\sfrac{1}{2}\epsilon_0 , \sfrac{1}{2} \epsilon_0 , \sfrac{1}{2} \epsilon_0 ) $ as a function of $\beta$.
Now $e_{(n,m)}[w(\beta)]$ and $\psi_{(n,m)}[w(\beta)]$ are   functions of $\beta$ as well as their limits as $m$ tends to infinity.
First we  show that $e_{(n,m)}[w(\beta)] \to e_{(n,\infty)}[w(\beta)]$ converges uniformly in $C^k(\R)$ for any $n$. This will then imply that $e_{(n,\infty)}$ is
in $C^k$.
We introduce for $\gamma, \delta   > 0$ the balls
$$
\mathcal{E}(\gamma, \delta) := \{ f \in C^k(\R;\C) |  \| f \|_\infty < \gamma  ,  \max_{1 \leq l \leq k} \| \partial_1^l  f \|_\infty < \delta \} .
$$
Let $w \in \mathcal{B}^{(k)}(\cdot , \epsilon , \cdot)$ with $\epsilon \leq \rho/32$.
We define a mapping $K[w]$ on $\mathcal{E}(1/2,\delta) $   by
$$
(K[w](f))(\beta) := I_\rho(\beta, f(\beta)).
$$
>From Lemma  \ref{renorm:thm3} it follows that  $K[w] ( \mathcal{E}(1/2,\delta)) \subset \mathcal{E}(3/8,\infty)$.
Using Faa di Bruno's formula we find
\begin{align*}
D_\beta^s I_\rho(\beta, f (\beta)) &= \sum_{p=0}^s \binom{s}{p}
\sum_{X \in P_p} \partial_1^{s-p} \partial_2^{|X|} I_\rho(\beta, f (\beta)) \prod_{x \in X} \partial_\beta^{|x|} f(\beta) .
\end{align*}
We use this to estimate the following difference
\begin{eqnarray}
\lefteqn{ D_\beta^s I_\rho(\beta, f (\beta)) - D_\beta^s I_\rho(\beta, g (\beta)) } \nonumber \\
&& = \Sigma'  \partial_1^{s-p} \partial_2^{|X|} \left[ I_\rho(\beta, f (\beta))  -   I_\rho(\beta, g (\beta)) \right]
  \prod_{x \in X} \partial_\beta^{|x|} f(\beta) \nonumber \\
&& + \Sigma' \partial_1^{s-p} \partial_2^{|X|} I_\rho(\beta, g (\beta)) \left[
 \prod_{x \in X} \partial_\beta^{|x|} f(\beta) -    \prod_{x \in X} \partial_\beta^{|x|} g(\beta)  \right] ,  \label{eq:faadibruno888}
\end{eqnarray}
where we used the abbreviation $\Sigma' =  \sum_{p=0}^s \binom{s}{p} \sum_{X \in P_p}$.
To estimate \eqref{eq:faadibruno888} we use that
\begin{equation} \label{eq:diff33333}
\left| I_\rho(\beta, f (\beta))  -   I_\rho(\beta, g (\beta)) \right| \leq
 \sup_{z \in D_{1/2} }| \partial_2 I_\rho(\beta,z ) | | f (\beta) - g(\beta) |
\end{equation}
and that for  $f,g \in \mathcal{E}(1/2,1)$ we have
\begin{equation}
\left|  \prod_{x \in X} \partial_\beta^{|x|} g(\beta)  -  \prod_{x \in X} \partial_\beta^{|x|} f(\beta) \right|  \leq  C_{|X|}  \| f - g \|_{C^k(\R)} ,
\end{equation}
for some constant depending only on the number of elements of the partition $X$.
On the other hand by Lemma \ref{lem:estonbderofinverse}  there exists a constant $C$
such that for all $(\beta,z) \in \R\times D_{3/8}$, we have
\begin{equation} \label{eq:estonbetader888}
\max_{1 \leq  l'  \leq k + 1} | \partial_z^{l'} I_\rho(\beta,z) | \leq C  \rho , \quad
\max_{1 \leq  l \leq k} \max_{0 \leq  l'  \leq k + 1}  |\partial_\beta^l \partial_z^{l'} I_\rho(\beta,z) | \leq {C}\epsilon  ,
\end{equation}
where we used the analyticity of $I_\rho$ in its second argument.
Using   \eqref{eq:estonbetader888}--\eqref{eq:diff33333}  to estimate  \eqref{eq:faadibruno888}
 it follows that for $\epsilon$ and $\rho$ sufficiently small we have
\begin{equation}  \label{eq:estonbetader999}
K[w] (\mathcal{E}(3/8,1)) \subset  \mathcal{E}(3/8,1) , \quad \| K[w] f - K[w] g \|_{C^k(\R)}  \leq \frac{1}{2}  \| f - g \|_{C^k(\R)}
\end{equation}
 for all $f,g \in  \mathcal{E}(3/8,1)$.
For the sequence of kernels $w^{(l)} \in  \mathcal{B}_0^{(k)}(\cdot , 2^{-l} \epsilon_0 , \cdot)$
define  $K_l := K[w^{(l)}]$. By definition we have
$$
e_{(n,m)} = K_n \circ K_{n+1} \circ \cdots \circ K_m(0) ,
$$
where $0$ denotes the zero function.
Thus if we choose $\rho$ and $\epsilon_0$ sufficiently small, then it follows from  \eqref{eq:estonbetader999}
that
$$
\| e_{(n,m)} -  e_{(n,m+l)} \|_{C^k(\R)} \leq 2^{-(m-n)-1} ,
$$
and thus $e_{(n,m)} \to e_{(n,\infty)}$ uniformly  in $C^k(\R)$ as $m \to \infty$ for any $n$.
Since $e_{(n,n)} = 0$ it follows that
\begin{equation} \label{eq:ckboundone}
\| e_{(n,\infty)} \|_{C^k(\R)} \leq 2 .
\end{equation}

Thus $e_{(n,m)}[w(\beta)] \to e_{(n,\infty)}[w(\beta)]$ converges uniformly in $C^k(\R)$ for any $n$.
Next we show that the groundstate eigenvector $\psi_{(0,\infty)}[w(\beta)]$ is $C^k$ in $\beta$.
For notational compactness we write $\psi_{(n,m)}(\beta)$ for $\psi_{(n,m)}[w(\beta)]$ and similarly
$e_{(n,m)}(\beta)$ for $e_{(n,m)}[w(\beta)]$.
We set    $\tilde{W}_n(\beta,z) := W[w^{(n)}(\beta,z)]$ with $w^{(n)}(\beta,z) = w^{(n)}(\beta)(z)$.
Observe that with this notation
${W}_n(\beta) := W_n[w(\beta)]  = \tilde{W}_n(\beta,e_{(0,\infty)}(\beta))$.   We use analogous definitions for   ${T}_n$, ${W}_n$, and ${Q}_n$.
Let $(\beta,z) \in \R \times D_{1/2}$.
We estimate the derivatives with respect to $\beta$
of
$$
\psi_{(n,m+1)} - \psi_{(n,m)} = Q_n \Gamma_\rho^* Q_{n+1} \cdots Q_{m-1} \Gamma^*_\rho ( Q_m - \chi_\rho) \Omega .
$$
Let
$$
A_n := \left. \left( T_n + \chib_\rho W_n \chib_\rho \right) \right|_{\ran \chib_\rho} .
$$
Observe that
\begin{equation} \label{eq:estoninvAn}
\| A_n^{-1} \| \leq 16/\rho .
\end{equation}
This can be seen using  
$\| W_n \| \leq 2^{-n-1} \epsilon_0 \leq \rho/16$,
see  \cite{BCFS03,HH10-1} for details.
We have already proved estimates of the form
$$
| \partial_\beta^l e_{(n,\infty)} | \leq c_l
$$
for $n \in \N_0$ which we will use without comment. We also have estimates of the form
\begin{align}
& \| w_{\geq 1}^{(n)} \|_\xi^{(k)} \leq \frac{\epsilon_0}{  2^n} \label{Junstar1} \\
& \| w_{0,0}^{(n)} \|^{(k)}  \leq \frac{\epsilon_0}{2^n} + \frac{1}{2} + \epsilon_0 + 1 \leq 2 \epsilon_0 + \frac{3}{2} . \label{Junstar2}
\end{align}
By the inequality given in Theorem \ref{thm:injective} and the differentiability of the integral kernels
it follows that $T_n$ and $W_n$ are differentiable functions of  $\beta$ (w.r.t. the operator norm topology) with
uniformly bounded derivatives.
And hence also $Q_n$ and $\psi_{(n,m)}$.
We have
$$
D_\beta^l ( Q_n - \chi_\rho) = - \sum_{l_1 + l_2 = l} \frac{l!}{l_1! l_2!} \chib_\rho \left[ D_\beta^{l_1} A_n^{-1} \right]
\chib_\rho D_\beta^{l_2} W_n \chi_\rho .
$$
It is straight forward to verify that for all $l \leq k$,
$$
\|  D_\beta^{l} A_n^{-1}  \| \leq C .
$$
To see this we note that taking inverses is a differentiable mapping with respect to the operator norm topology,
the first $k$ derivatives of $T_n$ and $W_n$ with respect to $\beta$ are uniformly bounded, and \eqref{eq:estoninvAn}.
Since
$$
\left. D_\beta W_n  \right\vert_{\beta}= \left. \left( \frac{\partial \tilde{W}_n}{\partial \beta} + \partial_z \tilde{W}_n \partial_\beta e_{(n,\infty)}  \right) \right\vert_{(\beta, e_{(n,\infty)}(\beta))}
$$
it is clear that if we can show that for $l,l' \leq k$
\begin{equation} \label{Junstar3}
\| \partial_\beta^l \partial_z^{l'} \tilde{W}_n (\beta,e_{(n,\infty)}(\beta))\| \leq \frac{c_l}{2^n} ,
\end{equation}
it will follow that for $l \leq k$,
\begin{equation} \label{eq:estonQ888}
\| D_\beta^l( Q_n - \chi_\rho) \| \leq \frac{c_l}{2^n} .
\end{equation}
The Cauchy integral formula gives
$$
\partial_\beta^l \partial_z^{l'} \tilde{W}_n(\beta,z)
= \frac{l'!}{2 \pi i} \int_{|\zeta| = 1/2 - \epsilon }  \frac{\partial_\beta^l \tilde{W}_n(\beta,\zeta )}{(\zeta - z)^{l'+1}}  d\zeta
$$
If $|z| < 1/2 - \epsilon $. Since $e_{(n,\infty)} \in D_{5 \rho /8}$ we obtain from \eqref{Junstar1}
$$
\left\|( \partial_\beta^l \partial_z^{l'} \tilde{W}_n)(\beta, e_{(n,\infty)}(\beta)) \right\| \leq \frac{(l')! {\epsilon_0}}{{2^{n+1}}(1/2 - 5 \rho / 8 )^{l'+1} } \leq
\frac{c}{2^n} .
$$
Thus we have shown \eqref{eq:estonQ888}. Using this inequality we find for $l \leq k$, with $p = m-n+1$,
\begin{align} \label{eq:ckboundonpsi0}
 \| D_\beta^l ( \psi_{(n,m+1)} - \psi_{(n,m)} ) \| &= \|   D_\beta^s  Q_{n} \Gamma_\rho^* Q_{n+1} \cdots Q_{m-1} \Gamma_\rho^* ( Q_{m} - \chi_\rho ) \Omega \| \\
&= \sum_{\underline{l} \in \N_0^{p} : | \underline{l} | = l } \frac{l!}{\underline{l}!}
 (D_\beta^{l_{1}} Q_n) \Gamma_\rho^* \cdots ( D_{\beta}^{l_{p-1}} Q_{m-1} ) \Gamma_\rho^* ( D_{\beta}^{l_p}  ( Q_{m} - \chi_{\rho} ) ) \Omega \nonumber \\
& \leq (m-n+1)^l \prod_{j=n}^{m-1} \left( 1 + \frac{C}{2^{j}} \right) \frac{C}{2^{m}} \leq (m+1)^{l} C 2^{-m} \exp(  C \sum_{j=1}^\infty 2^{-j}) . \nonumber
\end{align}
 \label{eq:estonQ111}
This implies that  $\psi_{(n,m)}[w(\beta)] \to \psi_{(n,\infty)}[w(\beta)]$ converges uniformly in $C^k(\R)$ for any $n$.
Since $\psi_{(n,n)} = \Omega$, it follows that
\begin{equation} \label{eq:ckboundonpsi1}
 \| \psi_{(n,\infty)} \|_{C^k(\R)} \leq 1 + C e^{2C} \sum_{m=0}^\infty (m+1)^k 2^{-m} .
\end{equation}
Now (d) follows.

 To show (e) first observe  by Theorem
\ref{thm:analytbasic2}  $(s,z,\beta) \mapsto w^{(n)}(s,z,\beta) = \mathcal{R}_\rho^n(w(s,\beta))(z)$  is
in $C^{\omega,k}_B(S \times D_{1/2} \times \R ; \WW_\xi^\#)$. It follows by  \eqref{eq:pzeestimate}
that $J_n^{-1}  \in C_B^{\omega,k}(S \times D_{1/2} \times \R )$. Thus
$e_{(n,m)} \in C_B^{\omega,k}(S \times \R) \cong  C_B^{\omega}(S ; C^k_B( \R))$.
It follows from the uniform convergence established in (d) that $e_{(n,\infty)} \in C^\omega_B(S ; C_B^k(\R))$.
It now follows from the bound in Theorem \ref{thm:injective} and the chain rule that
$H_n[w], W_n[w]$  are in $C_B^{\omega,k}(S \times \R ; \mathcal{B}(\HH_{\rm red}))$. Since
$H_n[w]$ is bounded invertible on the range of $\chib_\rho$ it follows from the bound \eqref{eq:estoninvAn} that $Q_n[w] \in C_B^{\omega,k}(S \times \R ; \mathcal{B}(\HH_{\rm red}))$.
Thus $\psi_{(n,m)} \in C_B^{\omega,k}(S \times \R ; \HH_{\rm red})
\cong C_B^\omega(S ; C^k(\R ; \HH_{\rm red} ))$. By the uniform convergence established in
\eqref{eq:ckboundonpsi0} it follows that $\psi_{(n,\infty)} \in C_B^\omega(S ; C_B^k(\R ; \HH_{\rm red}))$.

\section{Contraction Estimate}

\label{app:contraction}

In this section we prove  Theorem \ref{contk1}.
 By Lemma \ref{codim:thm3first} we know that there exists a constant $C_\theta$ which
is greater than 1 such that for   $w \in \mathcal{B}^{(\#, k)}(\rho/32,5 \rho/8,\rho/32)$. We have
\begin{equation} \label{eq:vestimate22}
 \max_{0 \leq l \leq k} || \partial_\beta^l  v_{\umm,\upp,\unn,\uqq}[w(\beta)] ||^{\#}   \leq
 C_\theta \left(\frac{16}{\rho} \right)^{L-1}   \prod_{l=1}^{L}
    \frac{ \max_{0 \leq l' \leq k} \| \partial_\beta^{l'}w_{m_l + p_l, n_l + q_l}(\beta) \|^{\#} }  { \sqrt{p_l ! q_l !}}  .
\end{equation}
The crucial point of equation  \eqref{eq:vestimate22}  is that  $\rho^{-1}$ occurs to
a power of at most $L-1$.
This allows us to prove Theorem \ref{contk1} using  similar  estimates as the proof of Theorem 38  \cite{HH10-1}, or Theorem 3.8 in \cite{BCFS03}.
There is an additional complication  due to the $\beta$ dependence of the reparameterization of
the spectral parameter.
We introduce the constant $D_k =  \sum_{l=0}^k \binom{k}{l} \sum_{X \in P_l}   1 $.

\vspace{0.5cm}

\noindent
Let   $0 < \rho \leq  ( k! 16 C_\theta D_k C_k^k  )^{-1} $,  $0 < \xi \leq
\min(1/2, (\frac{C_\theta}{64}  \tau
C_k^k D_k)^{-1/4}) $, and $0 < \epsilon_0 \leq  \min( \frac{\rho}{32},  \frac{1}{D_k 8^{k+1} k!  C_k^k})              $.

\vspace{0.5cm}

\noindent
We assume that  $w \in \mathcal{B}^{(k)}(\epsilon, \delta_1,\delta_2)$ with
 $\epsilon,\delta_1,\delta_2 \in [0, \epsilon_0 )$. Then
 the following estimates hold.

\vspace{0.5cm}

\noindent
\underline{Step 1:}  We have
$$
\| \mathcal{R}_\rho(w)_{\geq 2}  \|_\xi^{(k)} \leq    \frac{1}{2}   \| w_{\geq 2} \|_\xi^{(k)} .
$$

\vspace{0.5cm}
By definition $(\mathcal{R}_\rho w)(\beta,z) = \mathcal{R}_\rho^\#(w(\beta,I_\rho(\beta,z))$.
Taking the derivative with
respect to $\beta$ we obtain
\begin{eqnarray}
D_\beta^l ( \mathcal{R}_\rho w)(\beta,z) && =
 \partial_{\beta}^{l}  \mathcal{R}_\rho^\# (w(\beta,\zeta) )  \Big\vert_{\zeta = I_\rho(\beta,z)}  \nonumber  \\
&&
+ \sum_{ p=1}^{l} \binom{l}{p} \sum_{X \in P_p} \partial_{\beta}^{l-p} \partial_{\zeta}^{|X|} \mathcal{R}_\rho^\# (w(\beta,\zeta))
\prod_{x \in X } \partial_\beta^{|x|} I_\rho(\beta,z)  \Big\vert_{\zeta = I_\rho(\beta,z)} . \label{eq:faaonrenorm1}
\end{eqnarray}
Let us first estimate the first term on the right hand side. To this end let $u \in D_{19 \rho/32}$. Then
$w(\beta,u)  \in \mathcal{B}^{(\#, k)}(\rho/32,5 \rho/8,\rho/32)$ as the following estimate shows,
$$
\|w_{0,0}(\cdot, u) \|_{C^k(\R)}  \leq \| w_{0,0}(\cdot, u)  + u \|_{C^k(\R)}  + |u |
\leq 5 \rho /8 .
$$
By \eqref{eq:renormalizedkernels01}   we find for $M + N \geq 2$,
\begin{eqnarray*}
\lefteqn{ \| \partial_\beta^l  \mathcal{R}_\rho^\#(w(\beta,u))_{M,N}  \|^{\#} } \\
&\leq & \sum_{L=1}^\infty \sum_{ \substack{ (\umm,\upp,\unn,\uqq) \in \N_0^{4L}:  \\  |\umm| = M  , | \unn| = N  , m_l+p_l+n_l+q_l \geq 1 }}
 \rho^{|\umm|+|\unn|-1}
      \prod_{l=1}^L     \binom{ m_l + p_l }{ p_l } \binom{ n_l + q_l }{  q_l }    \| \partial_\beta^l v_{\umm,\upp,\unn,\uqq}[w(\beta,u)] \|^{\#} .  \nonumber \\
 \end{eqnarray*}
Inserting this below and using  \eqref{eq:vestimate22}, we find with $\tau := 16/\rho$,
\begin{eqnarray}
\lefteqn{  \left\| \partial_\beta^l \left( \mathcal{R}_\rho^\#({w}(\beta,u) \right)_{\geq 2} \right\|_\xi^{\#} } \nonumber \\
&& =  \sum_{M+N \geq 2} \xi^{-(M+N)}  \max_{0 \leq l \leq k}   \| \partial_\beta^l  \mathcal{R}_\rho^\#(w(\beta,u))_{M,N}  \|^{\#}            \nonumber \\
&&\leq   \sum_{L=1}^\infty          \sum_{ \substack{ (\umm,\upp,\unn,\uqq)  \in \N_0^{4L}: \\
   |\umm| + | \unn| \geq 2 , m_l+p_l+n_l+q_l \geq 1} }
\rho^{-1}  \left(  2 \rho\right)^{|\umm|+|\unn|}   (2 \xi)^{-(|\umm|+|\unn|)}  C_\theta \tau^{L-1} \nonumber \\
 && \quad \times
      \prod_{l=1}^L  \left\{   \binom{ m_l + p_l }{p_l } \binom{ n_l + q_l}{  q_l }
     \frac{  \max_{0 \leq l' \leq k}     \| \partial_\beta^{l'}  w_{m_l + p_l, n_l + q_l}(\beta,u)   \|^{\#} }{\sqrt{p_l ! q_l !}}  \right\}  \nonumber \\
   &&\leq
   \frac{C_\theta}{16} [2 \rho]^2   \sum_{L=1}^\infty  \tau^{L}
     \sum_{  \substack{  (\umm,\upp,\unn,\uqq) \in \N_0^{4L}:  \\
   m_l+p_l+n_l+q_l \geq 1 }} \nonumber \\
   && \times
     \prod_{l=1}^L \left\{   \binom{ m_l + p_l}{ p_l } \binom{ n_l + q_l }{ q_l }      \xi^{p_l + q_l} 2^{-(m_l+n_l)}
                                        \xi^{-(m_l + p_l +n_l + q_l)}
 \max_{0 \leq l' \leq k}     \| \partial_\beta^{l'}  w_{m_l + p_l, n_l + q_l}(\beta,u)   \|^{\#}
  \right\} \nonumber \\
&& \leq \frac{C_\theta}{4} \rho^2
    \sum_{L=1}^\infty   \tau^L
    \left[   \sum_{  m+p+n+q \geq 1 }    \binom{ m + p }{ p } \binom{ n + q}{  q }      \xi^{p + q} 2^{-(m+n)}
                                        \xi^{-(m + p +n + q)}
 \max_{0 \leq l \leq k}     \| \partial_\beta^{l}  w_{m + p, n + q}(\beta,u)   \|^{\#}
 \right]^L \nonumber \\
 &&\leq \frac{C_\theta}{4} \rho^2
    \sum_{L=1}^\infty  \tau^L
    \left[   \sum_{  l+k \geq 1 }
   \xi^{-(l + k )}
 \max_{0 \leq l' \leq k}     \| \partial_\beta^{l'}  w_{l,k}(\beta,u)   \|^{\#}
\right]^L \nonumber \\
  && \leq
    \frac{C_\theta}{4} \rho^2     \sum_{L=1}^\infty  \tau^{L}     \left(  \| w_{\geq 2} \|^{(k)}_\xi \right)^L
   \nonumber \\
   && \leq
    8 C_\theta \rho \| w_{\geq 2} \|_\xi^{(k)}  \label{eq:estrenwhbeta},
\end{eqnarray}
where in the third last inequality we used the binomial formula and $0 < \xi \leq 1/2$  and we used
$ \tau \| w_{\geq 2} \|_\xi^{(k)} \leq 1/2$ in the last inequality.
Now we estimate the terms involving derivatives with respect to $\zeta$. By Cauchy we have for
$\zeta \in U[w] \subset D_{17 \rho/32}$
$$
\partial_\beta^l \partial_\zeta^s \mathcal{R}_\rho^\#( w(\beta, \zeta))
= \frac{s!}{2 \pi i} \int_{|\mu|=18\rho/32} \frac{ \partial_\beta^l \mathcal{R}_\rho^\#(w(\beta,\mu))}{( \mu - \zeta )^{s+1} } d \mu .
$$
Using this and  \eqref{eq:estrenwhbeta}, we obtain the bound
$$
\| \partial_\zeta^s ( \mathcal{R}_\rho^\# w )_{\geq 2} \|_\xi^{(k)} \leq
\left( \frac{32}{\rho} \right)^s s! 8    C_\theta \rho \| w_{\geq 2} \|_\xi^{(k)}
$$
Now by Lemma \ref{lem:estonbderofinverse} we know that for $1 \leq l \leq k$ there exists a finite constant
$C_k$ such that
$$
\sup_{(\beta,z) \in \R \times D_{1/2}} | \partial_\beta^l I_\rho(\beta,z) | \leq C_k \frac{\rho}{32} .
$$
This and \eqref{eq:faaonrenorm1} imply that the $\rho$'s cancel out. Collecting the above estimates we
arrive at the bound
$$
\| (\mathcal{R}_\rho w)_{\geq 2} \|_\xi^{(k)} \leq   k! 8 C_\theta \rho D_k C_k^k
\| w_{\geq 2} \|_\xi^{(k)} .
$$

\noindent
\underline{Step 2:}
\begin{eqnarray*}
\lefteqn{ \sup_{z \in D_{1/2}}\| \partial_r (\mathcal{R}_\rho w)_{0,0}(z) - 1 \|_{C^{k}(\R ; C_B[0,1])} } \\
&& \leq
 \sup_{z \in D_{1/2} } \| \partial_r w_{0,0}(z) - 1 \|_{C^{k}(\R; C_B[0,1])}  +
\frac{1}{4}  \sup_{z \in D_{1/2} } \| w_{0,0}(0,z) + z \|_{C^k(\R)} +
 \frac{1}{4}  \| w_{\geq 1} \|_\xi^{( k)} .
\end{eqnarray*}

By  \eqref{eq:renormalizedkernels00}   we have
\begin{align} \label{eq:eqstep2est}
\partial_r (\mathcal{R}_\rho w(\beta))_{0,0}(z,r) - 1 = (\partial_r w_{0,0})(\beta, I_\rho(\beta, z),\rho r ) - 1 +   \partial_r T[w(\beta,I_\rho(\beta,z)](r) ,
\end{align}
where we defined
$$
T[w] := \rho^{-1} (-1)^{L-1} \sum_{L=2}^\infty
\sum_{\substack{ (\upp,\uqq) \in \N_0^{2L}: \\ p_l + q_l \geq  1 } }  v_{\underline{0},\upp,
\uzz,\uqq}[w]  .
$$
We need to estimate the derivative with respect to $\beta$.  For the
first term in \eqref{eq:eqstep2est} we find for $1 \leq l \leq k$,
using \eqref{eq:faadibruno}
\begin{eqnarray} \label{eq:estonmultder55-1}
\lefteqn{
D_\beta^l (\partial_r w_{0,0})(\beta,I_\rho(\beta,z), \rho r)  }   \\
 && =
\partial_\beta^l (\partial_r w_{0,0} )( \beta, I_\rho(\beta,z) , \rho r )  \nonumber \\
&&
+ \sum_{p=1}^{l} \binom{l}{p} \sum_{X \in P_p} \partial_\beta^{l-p} \partial_{\zeta}^{|X|}
(\partial_r w_{0,0} )(\beta,\zeta, \rho r ) \vert_{\zeta  =  I_\rho(\beta,z)  } \prod_{x \in X} \partial_\beta^{|x|} I_\rho(\beta,z) . \nonumber
\end{eqnarray}
We use analyticity,  Cauchy, and that $\zeta =  I_\rho(\beta,z)  \in D_{3/8}$
to estimate the derivatives with respect to the spectral
parameter. We
have
$$
\partial_\zeta^s \left( (\partial_r w_{0,0} )( \beta , \zeta , \rho r  ) - 1 \right) =
\lim_{\eta \downarrow 0} \frac{s!}{2\pi i}
\int_{|\mu|= 1/2 - \eta } \frac{ (\partial_r w_{0,0} )(\beta , \mu , \rho r  )  - 1 }{( \mu - \zeta )^{s+1}} d \mu .
$$
This yields for  $1 \leq l \leq k$ or $0 \leq l \leq k$ and $1 \leq s$,
\begin{equation} \label{eq:estonmultder55-2}
| \partial_\beta^l \partial_\zeta^s (\partial_r w_{0,0} )(\beta , \zeta  , \rho r   ) | \leq  8^s s! a  ,
\quad  \forall \zeta \in D_{3/8},
\end{equation}
where $a :=  \| \partial_r w_{0,0} - 1 \|_{C^k(\R ; C_B[0,1])}$
Using estimate \eqref{eq:estonmultder55-2} and the estimate of Lemma \ref{lem:estonbderofinverse}
to bound the last line
 of \eqref{eq:estonmultder55-1}
we find
\begin{eqnarray}  \label{eq:estonmultder55-22}
\lefteqn{
\sup_{0 \leq l \leq k} | D_\beta^l ( (\partial_r w_{0,0})(\beta,I_\rho(\beta,z), \rho r) - 1 ) | }
\\
&&
\leq
   \| \partial_r w_{0,0} - 1 \|_{C^k(\R ; C_B[0,1])} +
 D_k 8^k k!  C_k^k   a   \| w_{0,0}(0,z) + z \|_{C^k(\R)}
\nonumber .
\end{eqnarray}
The second term in    \eqref{eq:eqstep2est}  is estimated as follows.
For $u \in D_{19 \rho /32}$ and  $0 \leq l \leq k$ we estimate
\begin{eqnarray}
| \partial_\beta^l  \partial_r T[w(\beta,u)](r) |
&& \leq  \rho^{-1} \sum_{L=2}^{\infty} C_\theta \tau^{L-1}
\sum_{\substack{ (\upp,\uqq) \in \N_0^{2L}: \\ p_l + q_l \geq  2 }     }
\prod_{l=1}^L  \frac{ \max_{0 \leq l' \leq k} \| \partial_\beta^{l'} w_{p_l,q_l}(\beta,u) \|^{\#}  }{\sqrt{p_l! q_l!}}  \nonumber  \\
&& \leq  \frac{C_\theta}{16} \sum_{L=2}^{\infty}  \left[\tau \xi^2 \right]^{L}
 \left[ \sum_{p+q \geq 2}  \xi^{-(p+q)} \max_{0 \leq l' \leq l} \|\partial_\beta^{l'} w_{p,q}(\beta,u)  \|^{\#}
\right]^L   \nonumber \\
 && \leq  \frac{C_\theta}{16}  \xi^4 \sum_{L=2}^{\infty}    \left[ \tau  \| w_{\geq 2} \|^{(k)}_{\xi} \right]^{L} \nonumber \\
 && \leq     \frac{C_\theta}{16}  \xi^4  \tau     \| w_{\geq 2} \|^{(k)}_{\xi}
 \label{eq:kernelestimate2}
\end{eqnarray}
where in  the last estimate we used $    \tau \| w_{\geq 1} \|_{\xi}^{(k)}  \leq 1/2$.
Now using a contour estimate as in Step 1 one can show that
\begin{equation}   \label{eq:estonmultder55-3}
 \| D_\beta^l \partial_r T[(w(\beta,I_\rho(\beta,z))] \|_{C[0,1]} \leq  k! \frac{C_\theta}{16} \xi^4 \tau
C_k^k D_k \| w_{\geq 2} \|_\xi^{(k)} .
\end{equation}
Now estimates   \eqref{eq:estonmultder55-22} and   \eqref{eq:estonmultder55-3} yield Step 2.

\vspace{0.5cm}

\noindent
\underline{Step 3:}
\begin{align*}
\sup_{z \in D_{1/2}}\| (\mathcal{R}_\rho w )_{0,0}(z,0) + z \|_{C^k(\R)}
&  \leq  \frac{1}{4}  \| w_{\geq 1} \|_\xi^{(k)}  .
\end{align*}

 \vspace{0.5cm}

By  \eqref{eq:renormalizedkernels00}   we have
$$
(\mathcal{R}_\rho w(\beta) )_{0,0}(z,0) + z  = T[w(\beta,I_\rho(\beta,z))](0).
$$
We estimate for $u \in D_{19 \rho /32}$ and  $0 \leq l \leq k$ the same way as \eqref{eq:kernelestimate2}
\begin{align*}
| \partial_\beta^l T[ w(\beta,u)](0) | \leq 
 \frac{C_\theta}{16} \xi^4  \tau \| w_{\geq 1} \|_\xi^{(k)} .
\end{align*}
As above one calculates the derivative with respect to $\beta$ and estimates
the derivatives with respect to the spectral parameter using a contour integral
as in Step 1. As a result
\begin{align*}
\sup_{z \in D_{1/2}}\| (\mathcal{R}_\rho w )_{0,0}(z,0) + z \|_{C^k(\R)} &
\leq  k! \frac{C_\theta}{16} \xi^4 \tau
C_k^k D_k \| w_{\geq 2} \|_\xi^{(k)} .
\end{align*}
Step 3 now follows.

\section{Main Theorem}

\label{sec:prov}

In this section, we prove Theorem \ref{thm:main1}, the main result of this paper.
Its proof is based on Theorems \ref{thm:inimain1} and \ref{thm:bcfsmain}.

\vspace{0.5cm}

\noindent
{\it Proof of Theorem \ref{thm:main1}.} Choose $\rho, \xi, \epsilon_0$ such that the assertions of
Theorem \ref{thm:bcfsmain} hold.  Choose $g_0$ such that the conclusions
of Theorem \ref{thm:inimain1} hold for $\delta_1=\delta_2=\delta_3 = \epsilon_0/2$. Let   $ g \in D_{g_0}$. It  follows from
Theorem \ref{thm:bcfsmain} (a)
 that
$\psi_{(0)}[w^{(0)}(g,\beta)]$ is a nonzero element in the kernel of $H_{g,\beta}^{(0)}(e_{(0)}[w^{(0)}(g,\beta)])$.
 From the Feshbach property, Theorem  \ref{thm:fesh}, it follows that
\begin{equation} \label{eq:psiis}
\psi_{\beta}(g) := Q_{\chi^{(I)}}(g,\beta, e_{(0)}[w^{(0)}(g,\beta)]) \psi_{(0)}[w^{(0)}(g,\beta)]  ,
\end{equation}
is nonzero and an eigenvector of $H_{g,\beta}$ with eigenvalue $E_{\beta}(g) := E_{\rm at} + e_{(0)}[w^{(0)}(g,\beta)] $.

By Theorem \ref{thm:inimain1}, we know that $g \mapsto w^{(0)}(g,\cdot,\cdot)$ is an analytic $\WW_\xi^{(k)}$--valued function,
with values in  the ball $\mathcal{B}^{(k)}(\epsilon_0/2,\epsilon_0/2,\epsilon_0/2)$.
By Theorem \ref{thm:bcfsmain} (d) it follows that the functions
 $g \mapsto \psi_{(0)}[w^{(0)}(g,\cdot)] $ and $g \mapsto E_{(\cdot)}(g)$ are in
 $C^{\omega}_B(D_{g_0}; C_B^k(\R;\FF))$ and  $C^{\omega}_B(D_{g_0}; C_B^k(\R))$, respectively.
 From Theorem \ref{thm:inimain1} we know that the function $(g,z) \mapsto Q_{\chi^{(I)}}(g,\cdot,z)$ is in
$C^{\omega}_B(D_{g_0} \times D_{1/2} ; C^k_B(\R;\mathcal{B}(\HH_{\rm red};\HH) ))$.  It now follows
from \eqref{eq:psiis} that $g \mapsto \psi_{(\cdot)}(g)$ is in $C^{\omega}_B(D_{g_0}; C^k_B(\R ;\HH))$.
By possibly restricting to a smaller ball than $D_{g_0}$ we can ensure that the projection operator
\begin{equation} \label{eq:projectionlambda0}
P_{\beta}(g)
:= \frac{ \left| \psi_{\beta}(g) \right\rangle \left\langle \psi_{\beta}(\overline{g}) \right| }{ \left\langle \psi_{\beta}(\overline{g}) , \psi_{\beta}({g}) \right\rangle } ,
\end{equation}
is well defined for all $(g,  \beta) \in D_{g_0} \times \R$,
which is shown as follows. First observe that the denominator of \eqref{eq:projectionlambda0} is for each $\beta$
an analytic complex valued
function of $g$. By Theorem \ref{thm:bcfsmain} (c) we have
$\langle \psi_{\beta}(0) , \psi_{\beta}(0) \rangle = 1$. If we estimate
the remainder of the Taylor expansion of the denominator of \eqref{eq:projectionlambda0} using analyticity and
the uniform bound on $\psi_{(\cdot)}$, it follows, by possibly choosing $g_0$ smaller but still positive, that there exists a positive constant $c_0$ such that
$| \langle \psi_{\beta}(\overline{g}) , \psi_{\beta}({g}) \rangle | \geq c_0 $
for all $|g| \leq {g}_0$.
Using already established properties  of $\psi_{\beta}(g)$, it follows from \eqref{eq:projectionlambda0} that
$g \mapsto P_{(\cdot)}(g)$ is in $C^{\omega}_B(D_{g_0}; C_B^k(\R ; \mathcal{B}(\HH)))$.
If $g \in D_{g_0} \cap \R$, then by definition \eqref{eq:projectionlambda0} we see that
$P_{\beta}(g)^* = P_{\beta}(\overline{g})$.
The kernel
$w^{(0)}(g,\beta)$ is symmetric for $g \in D_{g_0} \cap \R$, see Theorem \ref{thm:inimain1}. Exactly the same
way as in the proof of Theorem 1 in \cite{HH10-2} one can show that
$E_{\beta}(g) = \inf \sigma(H_{g,\beta})$ for real $g \in D_{g_0} \cap \R$.
\qed

\vspace{0.5cm}

\noindent
{\it Proof of Corollary \ref{cor:main1}.}
We use Cauchy's formula.
 For any positive $r$ which is less than $g_0$, we have
\begin{equation} \label{eq:cauchy223}
E_{\beta}^{(n)} = \frac{1 }{2 \pi i} \int_{|z|= r} \frac{E_{\beta}(z)}{z^{n+1}} d z , \
\psi_{\beta}^{(n)} = \frac{1 }{2 \pi i} \int_{|z|= r} \frac{\psi_{\beta}(z)}{z^{n+1}} d z , \
P_{\beta}^{(n)} = \frac{1 }{2 \pi i} \int_{|z|= r} \frac{P_{\beta}(z)}{z^{n+1}} d z .
\end{equation}
The first equation of  \eqref{eq:cauchy223}
 implies  that $\beta \mapsto E_{\beta}^{(n)}$ is in $C_B^k(\R)$ and that
 $\|E_{(\cdot)}^{(n)}\|_{C^k(\R)} \leq r^{-n} \| E_{(\cdot)} \|_{C^{\omega}_B(D_{g_0};C_B^k(\R))}$.
Similarly we conclude by \eqref{eq:cauchy223} that
$ \psi_{\beta}^{(n)}$ and   $P_{\beta}^{(n)}$ are as functions of $\beta$ in $C^k_B(\R; \HH)$ and $C^k_B(\R; \mathcal{B} (\HH))$, respectively, and
that there exists a finite constant $C$ such that $\| \psi_{(\cdot)}^{(n)}\|_{C^k(\R;\HH)} \leq C r^{-n}   $ and
$\|P^{(n)}_{(\cdot)} \|_{C^k(\R;\mathcal{B}(\HH))} \leq C r^{-n}$.
Finally observe that $(-1)^N H_{g,\beta}(-1)^N = H_{-g,\beta}$ where $N$ is the  linear operator on $\FF$ with $N \upharpoonright \FF^{(n)}(\hh) = n$.
This implies that the ground state energy $E_{\beta}(g)$ cannot depend on odd powers of $g$.
\qed

\section*{Acknowledgements}

 D. H. wants to thank J. Fr\"ohlich, G.M. Graf, M. Griesemer, and A. Pizzo for interesting conversations.
Moreover, D.H. wants to thank ETH Z\"urich   for hospitality and financial support.
I.H. thanks the Bernoulli Institute for its hospitality and support.

\section*{Appendix A: Elementary Estimates and the Pull-through Formula}
\label{sec:appA}

To give a precise meaning to expressions which occur in \eqref{eq:defhmn11} and \eqref{eq:defhlinemn}, we introduce the following.
For $\psi \in \FF$ having finitely many particles we have
\beqn \label{eq:defofa}
\left[ a(K_1) \cdots a(K_m) \psi \right]_n(K_{m+1},...,K_{m+n}) = \sqrt{\frac{(m+n)!}{n!}} \psi_{m+n}(K_{1},...,K_{m+n}) ,
\eeqn
for all $K_1,...,K_{m+n} \in \underline{\R}^3 := \R^3 \times \Z_2$, and
using Fubini's theorem it is elementary to see that the vector valued map
 $(K_1,...,K_m) \mapsto a(K_1) \cdots a(K_m) \psi$ is
an element of $L^2((\underline{\R}^{3})^m; \FF)$. The following lemma states the well
known pull-through formula.
For a proof see for example \cite{BFS98,HH10-1}.
\begin{lemma} \label{lem:pullthrough}
Let $f : \R_+ \to \C$ be a bounded measurable function. Then for all $K \in \R^3 \times \Z_2$
$$
f(H_f) a^*(K) = a^*(K) f(H_f + \omega(K) ) , \quad a(K) f(H_f) = f(H_f + \omega(K) ) a(K) .
$$
\end{lemma}
Let $w_{m,n}$ be function on $\R_+ \times ({\underline{\R}^{3})}^{n+m}$ with values
in the linear operators of $\HH_{\rm at}$ or the complex numbers.
To such a function we associate the quadratic form
\begin{eqnarray*}
q_{w_{m,n}}(\varphi,\psi) := \int_{{(\underline{\R}^3)}^{m+n}} \frac{ dK^{(m,n)}}{|K^{(m,n)}|^{1/2}}
 \left\langle a(K^{(m)}) \varphi ,w_{m,n}(H_f, K^{(m,n)}) a(\widetilde{K}^{(n)}) \psi \right \rangle ,
\end{eqnarray*}
defined  for all $\varphi$ and $\psi$ in $\HH$ respectively $\FF$, for which the right hand side is defined as a complex number.
To associate an operator to the quadratic form we will use the following lemma.
\begin{lemma} \label{kernelopestimate} Let $\underline{X} = \R^3 \times \Z_2$. Then
\begin{eqnarray} \label{eq:defofH}
| q_{w_{m,n}}(\varphi,\psi) | \leq \| w_{m,n} \|_\sharp \| \varphi \| \| \psi \| ,
\end{eqnarray}
where
\begin{eqnarray*}
 \| w_{m,n} \|_\sharp^2 :=
\int_{\underline{X}^{m+n}} \frac{d K^{(m,n)}}{|K^{(m,n)}|^2} \sup_{r \geq 0} \left[ \|w_{m,n}(r,K^{(m,n)}) \|^2 \prod_{l=1}^m \left\{ r + \Sigma[K^{(l)}] \right\}
 \prod_{\widetilde{l}=1}^n \left\{ r + \Sigma[\widetilde{K}^{(\widetilde{l})}] \right\} \right] .
\end{eqnarray*}
\end{lemma}
\begin{proof} We set $P[K^{(n)}] := \prod_{l=1}^n ( H_f + \Sigma[K^l])^{1/2}$ and insert 1's to obtain the trivial identity
\begin{align*}
 | q_{w_{m,n}}(\varphi,\psi) | &=
\Bigg| \int_{\underline{X}^{m+n}} \frac{d K^{(m,n)}}{|K^{(m,n)}|}
 \Big\langle
 P[K^{(m)}] P[K^{(m)}]^{-1} |K^{(m)}|^{1/2} a(K^{(m)}) \varphi , w_{m,n}(H_f ,K^{(m,n)})
 \\
&
\times
P[\widetilde{K}^{(n)}] P[\widetilde{K}^{(n)}]^{-1} | \widetilde{K}^{(n)}|^{1/2} a(\widetilde{K}^{(n)}) \psi \Big\rangle \Bigg| .
\end{align*}
The lemma now follows using the Cauchy-Schwarz
inequality and the following well known identity for $n \geq 1$ and $\phi \in \FF$,
\begin{eqnarray}
\int_{\underline{X}^n}
d K^{(n)} | K^{(n)} | \left\| \prod_{l=1}^n \left[ H_f + \Sigma[K^{(l)}] \right]^{-1/2} a(K^{(n)}) \phi \right\|^2 = \| P_\Omega^\perp \phi \|^2 \label{eq:trivialA}  ,
\end{eqnarray}
where
$P_\Omega^\perp := | \Omega \rangle \langle \Omega |$.
A proof of \eqref{eq:trivialA}  can for example be found in \cite{HH10-1} Appendix A.
\end{proof}

Provided the form $q_{w_{m,n}}$ is densely defined and $ \| w_{m,n} \|_\sharp$ is a finite real number,
then the form $q_{w_{m,n}}$ determines uniquely a bounded
linear operator  $\underline{H}_{m,n}(w_{m,n})$ such that
 $$
q_{w_{m,n}}(\varphi,\psi ) = \langle \varphi,\underline{H}_{m,n}(w_{m,n}) \psi \rangle ,
$$
for all $\varphi, \psi$ in the form domain of $q_{w_{m,n}}$. Moreover,
$\| \underline{H}_{m,n}(w_{m,n}) \|_{} \leq \| w_{m,n} \|_\sharp$.
Using the pull-through formula and Lemma \ref{kernelopestimate} it is easy to see
that for $w^{(I)}$, defined in \eqref{defofwI}, with $m+n=1,2$, the form
$$
q^{(I)}_{m,n}(\varphi, \psi) := q_{w_{m,n}^{(I)}}( \varphi , (H_f + 1 )^{-\frac{1}{2}(m+n)} (-\Delta + 1 )^{-\frac{1}{2} \delta_{1,m+n}} \psi )
$$
is densely defined and bounded.
Thus we can associate a bounded linear operator $L_{m,n}^{(I)}$ such that
$q_{m,n}^{(I)}(\varphi, \psi) = \langle \varphi , L_{m,n}^{(I)} \psi \rangle$. This allows us to define
$$
\underline{H}_{m,n}(w_{m,n}^{(I)}) := L_{m,n}^{(I)} (H_f + 1 )^{\frac{1}{2}(m+n)} (-\Delta + 1 )^{\frac{1}{2}\delta_{1,m+n}}
$$
as an operator in $\HH$.

\section*{Appendix B: Smooth Feshbach Property}
\label{sec:smo}

In this appendix we follow \cite{BCFS03,GH08}.
We introduce the Feshbach map and its auxiliary operator and state basic isospectrality
properties.
Let $\chi$ and $\overline{\chi}$ be commuting, nonzero bounded operators, acting on a separable Hilbert space $\HH$
and satisfying $\chi^2 + \overline{\chi}^2=1$. A {\it Feshbach pair} $(H,T)$ for $\chi$ is a pair of
closed operators with the same domain,
$$
H,T : D(H) = D(T) \subset \HH \to \HH
$$
such that $H,T, W := H-T$, and the operators
\begin{align*}
&W_\chi := \chi W \chi , & &W_{\overline{\chi}} := \overline{\chi} W \chib \\
&H_\chi :=T + W_\chi , & &H_{\overline{\chi}} := T + W_{\chib} ,
\end{align*}
defined on $D(T)$ satisfy the following assumptions:
\begin{itemize}
\item[(a)] $\chi T \subset T \chi$ and $\chib T \subset T \chib$,
\item[(b)] $T, H_{\chib} : D(T) \cap \ran \chib \to \ran \chib$ are bijections with bounded inverse,
\item[(c)] $\chib H_{\chib}^{-1} \chib W \chi : D(T) \subset \HH \to \HH$ is a bounded operator.
\end{itemize}
\begin{remark} \label{rem:abuse} {\em
By abuse of notation we write $ H_{\chib}^{-1} \chib$ for $ \left( H_{\chib} \upharpoonright \ran \chib \right)^{-1} \chib$ and
likewise $ T^{-1} \chib$ for $ \left( T \upharpoonright \ran \chib \right)^{-1} \chib$. }
\end{remark}
We call an operator $A:D(A) \subset \HH \to \HH$ {\it bounded invertible} in a subspace $V \subset \HH$
($V$ not necessarily closed), if $A: D(A) \cap V \to V$ is a bijection with bounded inverse.
Given a Feshbach pair $(H,T)$ for $\chi$, the operator
\begin{align} \label{eq:defoffesh}
&F_\chi(H,T) := H_\chi - \chi W \chib H_{\chib}^{-1} \chib W \chi
\end{align}
on $D(T)$ is called the { \it Feshbach map of} $H$.
The auxiliary operator 
\begin{align} \label{eq:defofQ}
  Q_\chi := Q_\chi(H,T) := \chi - \chib H_{\chib}^{-1} \chib W \chi
\end{align}
is by conditions (a), (c), bounded, and $Q_\chi$ leaves $D(T)$ invariant. The Feshbach map is
isospectral in the sense of the following theorem.
\begin{theorem} \label{thm:fesh}
Let $(H,T)$ be a Feshbach pair for $\chi$ on a Hilbert space $\HH$. Then the following holds.
$\chi \ker H \subset \ker F_\chi(H,T)$ and $Q_\chi \ker F_\chi(H,T) \subset \ker H$. The mappings
\begin{align*}
\chi : \ker H \to \ker F_\chi(H,T) , \quad Q_\chi : \ker F_\chi(H,T) \to \ker H ,
\end{align*}
are linear isomoporhisms and inverse to each other. 
\end{theorem}

The proof of Theorem \ref{thm:fesh} can be found in \cite{BCFS03,GH08}. The next lemma
gives sufficient conditions for  two operators to be a Feshbach pair. It follows
from a Neumann expansion, \cite{GH08}.

\begin{lemma} \label{fesh:thm2}
Conditions {\rm (a), (b)}, and {\rm (c)} on Feshbach pairs are satisfied if:
\begin{itemize}
\item[(a')] $\chi T \subset T \chi$ and $\chib T \subset T \chib$,
\item[(b')] $T$ is bounded invertible in $\ran \chib$,
\item[(c')] $\| T^{-1} \chib W \chib \| < 1$, $\| \chib W T^{-1} \chib \| < 1$, and $T^{-1} \chib W \chi$ is a bounded operator.
\end{itemize}
\end{lemma}

\section*{Appendix C: Function spaces}
\label{sec:banextra}

Let $(X, \| \cdot \|_X)$ and $(Y, \| \cdot \|_Y)$ be Banach spaces.
By  $\mathcal{B}(X,Y)$ we denote the Banach space of bounded
linear operators from $X$ to $Y$. We set $\mathcal{B}(X) := \mathcal{B}(X,X)$.
Let $(M, \mu)$ be a measure space.  We say that a function
$f : M \to X$ is measurable if there exists a sequence  $( f_j )_{j \in \N_0}$ of simple functions from $M$ to $X$,
such that $\| f_j(m) - f(m) \|_X \to 0 $ as $j \to \infty$,  for a.e. $m \in M$.
We define  $L^\infty(M;X)$ to be  the Banach space of measurable functions from $M$ to $X$ with norm
$$
\| f \|_{L^\infty (M;X)}
 :=                                {\rm ess} \sup_{m \in M} \| f(m) \|_X    .
$$
Let $[a,b]$ be a closed interval of $\R$. For $p \in \N_0$ we define the space
$
C^{p}[a,b]
$
to be  the space of  functions   $f: (a,b)  \to \C $  such that  for all $q=0,...,p$  the partial
derivatives $\partial_1^q f$ exist and are uniformly continuous on bounded subsets of  $(a,b)$.
We define the norm
\begin{align}
\| f  \|_{C^{p}[a,b]} &:= \max_{0 \leq q \leq p} \sup_{r \in (a,b) }
 | \partial_r^q  f(r) |
\end{align}
By $C^{p}_B[a,b]$ we denote  the Banach space with norm $\| \cdot  \|_{C^{p}[a,b]}$ which consists
of elements in   $C^{p}[a,b]$ for which the norm $\| \cdot  \|_{C^{p}[a,b]}$ is finite.
We denote by $C^k(\R;X)$ the space of strongly (w.r.t the norm in $X$) $k$--times continuously differentiable functions.  The norm is given by
$$
\| f \|_{C^k(\R;X)} := \max_{0 \leq s \leq k} \sup_{x \in \R} \| \partial_x^s f(x) \|_X .
$$
Let $C^k_B(\R;X)$ denote the set of functions $f$ in  $C^k(\R;X)$ for which the norm  $\| f  \|_{C^k(\R;X)} $  is finite.
Let $U \subset \C^n$ be  a domain. We define the  space
$
C^\omega(U ; X  )
$
to consist of all strongly analytic functions  $f : U \to  X $. We define  the norm
$$
\| f \|_{ C^\omega(U ; X)  } := \sup_{z \in U} \| f(z) \|_{X} .
$$
By $C^{\omega}_B(U;X)$ we denote  the Banach space with norm $\| \cdot  \|_{C^{\omega}(U;X)}$ which consists
of elements in   $C^{\omega}(U;X)$ for which the norm $\| \cdot  \|_{C^{\omega}(U;X)}$ is finite.
We define the space
$
C^{\omega,k}(U \times \R ; X )
$
to consist of all functions $f : U \times \R  \to X$ such that all partial derivatives
$\partial_x^{{l}} \partial_{z_i}^{{t}} f$,  with   ${l} \in \N_0$,  $l \leq k$, $i=1,...,n$, and  $t = 0,1$,  exist  and are
continuous. We define the norm
$$
\| f \|_{ C^{\omega,k}(U \times \R ; X)  } := \sup_{z \in U} \max_{0 \leq l \leq k} \sup_{x \in \R}
\| \partial_x^l f(z,x) \|_{X} .
$$
By $C^{\omega,k}_B(U  \times \R ;X)$ we denote  the Banach space with norm
 $\| \cdot  \|_{C^{\omega}(U \times \R;X)}$ which consists
of elements in   $C^{\omega}(U \times \R ;X)$ for which the norm $\| \cdot  \|_{C^{\omega}(U;X)}$ is finite.
In the case where $X = \mathbb{C}$ we will drop  the $X$ dependence  in the notation.
We introduce the Polydiscs  $D_{\underline{r}} = \prod_{i=1}^n D_{r_i}$ with  $\underline{r} \in (0,\infty)^n$.

\begin{lemma} \label{lem:weakstronanalyt1}
We have the canonical isomorphism of Banach spaces
\begin{equation} \label{eq:weakstronanalyt1}
C^{\omega,k}_B(D_{\underline{r}} \times \R ; X  )     \cong      C^\omega_B(D_{\underline{r}} ; C_B^k(\R ; X )) .
\end{equation}
\end{lemma}
\begin{proof}
Let  $f \in C^{\omega,k}_B(D_{\underline{r}} \times \R ; X  )$. Then for every $x \in \R$ the function
$z \mapsto f(z,x)$ is analytic on $D_{\underline{r}}$ and bounded.
Thus for $\epsilon > 0 $ sufficiently small
$$
f(z,x) = \sum_{\unn} c_{\unn} (x) z^{\unn}
$$
with
$$
c_{\unn}(x) = \frac{1}{(2 \pi i)^n}
\int_{D_{\underline{r} - \underline{\epsilon} } }  \frac{ f(\zeta,x)}{\zeta^{\underline{n}+\underline{1}}} \prod_{j=1}^n d \zeta_{j} ,
$$
where the integral is a strong Riemann integral in $X$ and we used the notation $\underline{1} = (1,...,1)$ and $\underline{\epsilon} = \epsilon \underline{1}$.
It follows that $\| c_{\unn} \|_{C^k(\R; X )} \leq \prod_{j=1}^n r_j^{-n_j} \| f \|_{C^{\omega,k}(D_{\underline{r}}\times \R ; X)}$.
This implies that the function $\widehat{f} : z \mapsto f(z, \cdot )$ is in  $C^{\omega}_B(D_{\underline{r}} ; C_B^k(\R ; X ))$.
Moreover,
$$
\sup_{z \in D_{\underline{r}} } \| \widehat{f}(z ) \|_{ C_B^k(\R ; X )      } =  \sup_{z \in D_{\underline{r}} }
\max_{0 \leq l \leq k} \sup_{x \in \R} \| \partial_x^l f(z,x) \|_X =
\| f \|_{C^{\omega,k}(D_{\underline{r}}\times \R ; X )} .
$$
Now suppose $g \in C^{\omega}_B(D_{\underline{r}} ; C_B^k(\R ; X ))$. Then
$$
g(z) =
\sum_{\underline{n}} a_{\underline{n}} z^{\underline{n}}
$$
with
$$
a_{\underline{n}} =
\frac{1}{(2 \pi i)^n}
\int_{D_{\underline{r} - \underline{\epsilon} } }  \frac{ g(\zeta)}{\zeta^{\underline{n}+\underline{1}}} \prod_{j=1}^n d \zeta_{j} ,
$$
where the integral is a strong Riemann integral in  $C_B^k(\R ; X )$. It follows that
\begin{equation} \label{eq:estonderofg}
 \| a_{\underline{n}} \|_{C^k(\R;X)} \leq  \prod_{j=1}^n r_j^{-n_j}
 \| g \|_{C^\omega (D_{\underline{r}} ; C_B^k(\R ; X ) )} .
\end{equation}
We define
$$
\tilde{g}(x,z) := \sum_{\unn} a_{\unn}(x)  z^{\unn}  .
$$
It follows from \eqref{eq:estonderofg} that $\tilde{g} \in C^{\omega,k}_B(D_{\underline{r}} \times \R ;X )$. Moreover,
$$
\| \widetilde{g} \|_{C^{\omega,k}(D_{\underline{r}}\times \R ; X )}
=  \sup_{z \in D_{\underline{r}} }
\max_{0 \leq l \leq k} \sup_{x \in \R} \| \partial_x^l \widetilde{g}(z,x) \|_X
=  \sup_{z \in D_{\underline{r}} } \| {g}(z ) \|_{C^k(\R;X)}  .
$$
\end{proof}


\section*{Appendix D: Faa di Bruno's Formula}
\label{app:faadibruno}

Let $P_n$ denote the set of all partitions of $\{1,...,n\}$. Then
\begin{equation} \label{eq:faadibruno}
(f \circ g)^{(n)} = \sum_{X \in P_n}  f^{(|X|)} \circ g   \prod_{ x \in X} g^{(|x|)} ,
\end{equation}
where $|X |$ and $|x|$ stand for the cardinality of the sets $X$ and $x$, respectively.

\section*{Appendix E: Uniform Convergence}

Let $(s_0,\beta_0) \in  S \times \R$. Then for every $\epsilon > 0$ there is an open set
$U \subset  S \times \R$ containing $(s_0,\beta_0)$ such that
$$
\sup_{(\beta,s) \in U} \max_{0 \leq l \leq k} \| \partial_\beta^l w(\beta,s) -
\partial_\beta^l w(s_0,\beta_0) \|_\xi^\# < \epsilon .
$$
This implies
$$
\sup_{(\beta,s) \in U} \max_{0 \leq l \leq k} \| \partial_\beta^l w(\beta,s)_{m,n} \|^\#  \leq
\max_{0 \leq l \leq k} \|
\partial_\beta^l w(s_0,\beta_0)_{m,n} \|^\# +    \xi^{m+n} \epsilon     =:E_{m,n}
$$
By Lemma \ref{codim:thm3first},
\begin{equation} \label{eq:estonvons}
{\rm sup}_{(\beta, s) \in U} \max_{0 \leq l \leq k} \| \partial_\beta^l v_{\umm,\upp,\unn,\uqq}[w(\beta, s)] \|^\# \leq C_L t^{-L + 1} \prod_{l=1}^L \frac{E_{m_l + p_l , n_l q_l}}{\sqrt{p_l ! q_l !}} ,
\end{equation}
where we used the notation introduced in that lemma. We estimate
\begin{align}
\label{eq:rdiffstup1a}
& \sum_{M+N \geq 0} \sum_{L=1}^\infty
\sum_{\substack{ (\umm,\upp,\unn,\uqq) \in \N_0^{4L} \\ | \umm | = M , | \unn | = N \\ m_l + p _l + n_l + q_l \geq 1}} \xi^{-|\umm|-|\unn|} \rho^{|\umm|+|\unn|}  \\
& \times
\prod_{l=1}^L \left\{  \binom{m_l + p_l }{ p_l} \binom{ n_l + q_l }{ q_l }      \right\}
{\rm sup}_{(\beta, s) \in U} \max_{0 \leq l \leq k} \| \partial_\beta^l v_{\umm,\upp,\unn,\uqq}[w(\beta, s)] \|^\#  \nonumber  \\
&\leq  \sum_{L=1}^\infty C_L t^{1-L}  G^{L}  \nonumber ,
\end{align}
where we used Eq.  \eqref{eq:estonvons} and the definition
$$
G := \sum_{ m+p + n+q \geq 1 }  \binom{m+p }{ p} \binom{n+q }{ q} \xi^{p+q} (1/2)^{m+n} \xi^{-m-p-n-q}  \frac{E_{m+p, n+q}}{\sqrt{p! q!}} .
$$
Below we will show that
\begin{equation} \label{eq:boundonGa}
G \leq \| w( s_0, \cdot )_{\geq 1} \|_\xi^{(k,\#)} + \epsilon 16 e^4 .
\end{equation}
Since $t^{-1} G   < 1$
for $\epsilon$ sufficiently small
Inequalities \eqref{eq:boundonGa}
imply the  convergence of \eqref{eq:rdiffstup1a}, for small $\epsilon$.
To show  \eqref{eq:boundonGa}, we will use the following estimate
\begin{equation}
\sum_{m+p \geq 0} \binom{ m + p }{ p} \xi^p (1/2)^m \frac{1}{\sqrt{p!}}
\leq \sum_{m+p \geq 0}  \binom{ m + p }{ p} (1/4)^p (1/2)^m e^{8 \xi^2}
= 4e^{8 \xi^2}
 \leq 4 e^2 , \label{eq:combestimatedera}
\end{equation}
where in the first inequality we used the trivial estimate  $(16\xi^2)^{p}/{p!} \leq e^{16\xi^2}$.
Now  \eqref{eq:boundonGa} is seen by    inserting  the definition of $E_{m,n}$  into the definition of $G$.
This yields two terms, which one has  to estimate. The second term, involving $\epsilon$,
is estimated using   \eqref{eq:combestimatedera}, and
the first term, involving $w_{m,n}(s_0,\beta_0)$, is estimated using the binomial formula, i.e.,
\begin{eqnarray*} \lefteqn{
\sum_{ m+p + n+q \geq 1 }  \binom{m+p }{ p} \binom{n+q }{ q} \xi^{p+q} (1/2)^{m+n} \xi^{-m-p-n-q}
\max_{0 \leq l \leq k} \| \partial_\beta^{l} w(\beta_0, s_0)_{m+p,n+q} \|} \\
&&=
\sum_{ i +  j \geq 1 }  (\xi + 1/2)^i  (\xi + 1/2)^j
\xi^{-i-j}
\max_{0 \leq l \leq k} \| \partial_\beta^{l} w(\beta_0, s_0)_{i,j} \| .
\end{eqnarray*}

\section*{Appendix G: Differentiability}

\begin{lemma} \label{lem:diff}   \quad
\begin{itemize}
\item[(a)]
The  mapping
\begin{eqnarray*}
 \widetilde{v}_{\umm,\upp,\unn,\uqq}[\cdot ]   :
(\WW_\xi^\# )^L \times (\WW_{0,0}^\#)^{L+1} &\to& \WW_{|\umm|,|\unn|}^\# \\
(w_1,...,w_L,G_0,...,G_L) &\mapsto  & \widetilde{v}_{\umm,\upp,\unn,\uqq}[w_1,...,w_L,G_0,...,G_L]
\end{eqnarray*}
defined by
\begin{eqnarray}
\lefteqn{ \widetilde{v}_{\umm,\upp,\unn,\uqq}[  w_1,...,w_L,G_0,...,G_L    ]( r, K^{(|\umm|,|\unn|)}) :=  } \label{eq:defofvtildea}  \nonumber \\
&&
\left\langle \Omega  ,  G_0(H_f + \rho ( r + \widetilde{r}_0 ) ) \prod_{l=1}^L
\left\{  {W}_{p_l,q_l}^{m_l,n_l}[w_l](\rho(r+r_l), \rho K_l^{(m_l,n_l)} ) G_l( H_f + \rho (r  + \widetilde{r}_l ) ) \right\} \Omega
\right\rangle    \nonumber .
\end{eqnarray}
is continuous and  multilinear.
\item[(b)]  The following mapping is in $C^\infty$.
\begin{eqnarray*}
\{ t  \in \WW_{0,0}^\# | \inf_{r \in [\rho\frac{3}{4},1] } | t(r) | > \epsilon \}
&\to&  \WW^\#_{0,0}  \\
t  &\mapsto&  \frac{\chib_{\rho}^2}{t}
\end{eqnarray*}
 \end{itemize}
\end{lemma}

\begin{proof}
(a) Using \eqref{eq:babyestimate1}
we find
\begin{eqnarray*}
\lefteqn{ \esssup_{ K^{(|\umm|,|\unn|)}} \sup_{r \in [0,1]}| \widetilde{v}_{\umm,\upp,\unn,\uqq}[ w_1,...,w_L,G_0,...,G_L    ](r, K^{(|\umm|,|\unn|)})  | } \\
&&
\leq   \prod_{l=1}^{L}  \esssup_{K^{(m_l,n_l)}} \sup_{r \in [0,1]}  \| W_{p,q}^{m,n}[w](r,K^{(m_l,n_l)})\|_{\rm op} ,
    \prod_{l=0}^{L} \|  G_l \|_{C[0,1]} .
\end{eqnarray*}
To estimate the right hand side we use
\begin{eqnarray} \label{eq:ineqforW1append}
&&\esssup_{K^{(m,n)}  }  \sup_{r \in [0,1]}  \|   W_{p,q}^{m,n}[w](r,K^{(m,n)}) \|_{\rm op} \leq
\frac{ \|   w_{p+m,q+n}, \|_{L^\infty(B_1^{m+n};C[0,1])}  }{\sqrt{p! q!}}
\end{eqnarray}
Inequality  \eqref{eq:ineqforW1append} can be shown using Lemma   \ref{kernelopestimate}
and \eqref{eq:intofwKminus2}.
Next we calculate the derivative with respect to $r$. To this end first observe that
using Lemma   \ref{kernelopestimate} and dominated convergence one can show  that
for a.e. $K^{(m,n)}$ the partial derivative  $ \partial_r W^{m,n}_{p,q}[w](r,K^{(m,n)})$
exists and equals $W^{m,n}_{p,q}[ \partial_r w](r,K^{(m,n)})$.
Using Leibniz we obtain
\begin{eqnarray*}
\lefteqn{ \partial_r \widetilde{v}_{\umm,\upp,\unn,\uqq}[  w_1,...,w_L,G_0,...,G_L    ]( r, K^{(|\umm|,|\unn|)}) =  }  \\
&&
 \rho \sum_{j=1}^{2L+1}  \widetilde{v}_{\umm,\upp,\unn,\uqq}[  \partial_r^{\delta_{1,j}} w_1,...,
\partial_r^{\delta_{L,j}} w_L,\partial_r^{\delta_{L+1,j}} G_0,..., \partial_r^{\delta_{2L+1,j}} G_L    ]( r, K^{(|\umm|,|\unn|)}) .
\end{eqnarray*}
Using again \eqref{eq:babyestimate1} and  \eqref{eq:ineqforW1append} to estimate this we find
$$
\|  \widetilde{v}_{\umm,\upp,\unn,\uqq}[  w_1,...,w_L,G_0,...,G_L    ] \|^\#_{} \leq
\prod_{l=1}^L \| w_l \|^{\#}_\xi  \prod_{l=0}^{L} \| G_l \|^\# .
$$
This yields (a).

(b) It is straight forward to verify that  the mapping
$t  \mapsto {\chib_\rho^2}/{t}  $ is differentiable
with derivative $-{\chib_\rho^2}/t^2$, see \cite{HH10-1} Lemma 36 (b).
 Using the product rule one can now  show iteratively that the function is in $C^\infty$.
\end{proof}

\end{document}